\documentclass[preprint,3p,times]{elsarticle}
\usepackage{natbib}
\usepackage{amsmath} 
\usepackage{amsfonts} 
\usepackage{bm} 
\usepackage{calc} 
\usepackage{float}
\usepackage{subcaption} 
\usepackage{hyperref} 

\usepackage{svg}
\usepackage{siunitx}
\usepackage{amssymb,moreverb,pgfplots,tikz}
\pgfplotsset{compat=1.12}
\usepackage{rotating}
\usepackage{enumitem}
\usepackage{cleveref}

\definecolor{mycolor1}{rgb}{0.00000,0.44700,0.74100}%
\definecolor{mycolor2}{rgb}{0.92900,0.69400,0.12500}%
\definecolor{MatlabBlue}{rgb}    {0     , 0.4470, 0.7410}
\definecolor{MatlabRed}{rgb}     {0.8500, 0.3250, 0.0980}
\definecolor{MatlabYellow}{rgb}  {0.9290, 0.6940, 0.1250}
\definecolor{MatlabPurple}{rgb}  {0.4940, 0.1840, 0.5560}
\definecolor{MatlabGreen}{rgb}   {0.4660, 0.6740, 0.1880}
\definecolor{MatlabBabyBlue}{rgb}{0.3010, 0.7450, 0.9330}
\definecolor{MatlabGray}{rgb}{0.5, 0.5, 0.5}
\definecolor{MatlabLightGray}{rgb}{0.75, 0.75, 0.75}
\definecolor{MatlabBlack}{rgb}{0, 0, 0}

\definecolor{MatlabLightGray4}{rgb}{0.875, 0.875, 0.875}
\definecolor{MatlabLightGray3}{rgb}{0.85, 0.85, 0.85}
\definecolor{MatlabLightGray2}{rgb}{0.775, 0.775, 0.775}
\definecolor{MatlabLightGray1}{rgb}{0.7, 0.7, 0.7}
\definecolor{MatlabGray30}{rgb}{0.3, 0.3, 0.3}
\definecolor{MatlabGray40}{rgb}{0.4, 0.4, 0.4}
\definecolor{MatlabGray50}{rgb}{0.5, 0.5, 0.5}
\definecolor{MatlabGray60}{rgb}{0.6, 0.6, 0.6}
\definecolor{MatlabGray70}{rgb}{0.7, 0.7, 0.7}
\definecolor{MatlabGray80}{rgb}{0.8, 0.8, 0.8}
\definecolor{MatlabGray90}{rgb}{0.9, 0.9, 0.9}

\definecolor{Red}{rgb}{1 0 0}
\definecolor{Black}{rgb}{0 0 0}

\definecolor{myblue}{rgb}{0 0 1}

\newcommand{\tikzline}[1]{(\protect\tikz[baseline=-0.6ex,x=1pt,y=1pt]{ \protect\draw[#1,thick] [-] (0,0) -- (10,0);})}
\newcommand{\tikzdashedline}[1]{(\protect\tikz[baseline=-0.6ex,x=0.9pt,y=1pt]{ \protect\draw[#1,thick,dashed] [-] (0,0) -- (10,0);})}

\newcommand{\tikzdottedline}[1]{(\protect\tikz[baseline=-0.6ex,x=1pt,y=1pt]{ \protect\draw[#1,thick,dotted] [-] (0,0) -- (10,0);})}

\newcommand{\tikzcross}[1]{(\protect\tikz[baseline=-0.6ex,x=1pt,y=1pt]{ \protect\draw[color=#1, thick] (-2,-2) -- (2,2) (-2,2) -- (2,-2) })}

\DeclareRobustCommand\encircle[1]{%
\tikz[baseline=(X.base)] 
   \node (X) [draw, shape=circle, inner sep=0] {\strut #1};}

\newcommand{\norm}[1]{\left\lVert#1\right\rVert}
\newcommand{\abs}[1]{\left\lvert#1\right\rvert}
\newcommand{\RHinf}{$\mathcal{R}\mathcal{H}_\infty$ }

\hypersetup{pdfauthor={Name}} 

\usepackage{amsthm}
\newtheorem{definition}{Definition}

\newtheorem{problem}{Problem}
\newtheorem{lemma}{Lemma}

\newtheorem{theorem}{Theorem}

\newtheorem{remark}{Remark}

\usepackage{graphicx}

\usepackage{fancyhdr}

\fancypagestyle{firstpage}{
    \fancyhf{}
    \fancyhead[C]{%
        \small
        Koen Classens et al., \textit{Closed-loop Optimal Fault Detection for Uncertain Systems},\\
        \textit{Submitted for journal publication}, uploaded to ArXiv June 15, 2026
    }
    \fancyfoot[C]{\thepage} 
    
}

\begin{document}
    
\begin{titlepage}
    \title{Closed-loop Optimal Fault Detection for Uncertain Systems}

    \author[1]{Koen Classens\corref{cor1}}
    \ead{k.h.j.classens@gmail.com}
    \author[1]{Tjeerd Ickenroth}
    \author[2]{Jeroen van de Wijdeven}
    \author[1,3]{Tom Oomen}
       
    \cortext[cor1]{Corresponding author}
   
    \address[1]{Eindhoven University of Technology, Department of Mechanical Engineering, Control Systems Technology, Eindhoven, The Netherlands}
    \address[2]{ASML, Veldhoven, The Netherlands}
    \address[3]{Delft University of Technology, Department of Mechanical Engineering, Delft Center for Systems and Control, Delft, The Netherlands}

    \begin{abstract}
    Faults compromise the reliability and safety of complex engineering systems. 
    The aim of this article is to address the problem of robust fault detection filter design for continuous-time linear time-invariant uncertain systems in open-loop or closed-loop configurations. 
    The developed method offers a unified approach to handle parametric and dynamic uncertainties by solving a single Riccati equation, based on a worst-case disturbance and uncertainty model. This worst-case model is obtained by nonlinear optimization and application of the boundary Nevanlinna-Pick method.
    The efficacy of the proposed approach is demonstrated using an uncertain model of an experimental reticle stage used in the lithography industry. The results illustrate that an optimal compromise is achieved between sensitivity to faults and rejection of modelling uncertainties and disturbances on the other hand. 
    This capability enables the clear differentiation between faults and undesired effects in residuals, thereby enhancing fault detection reliability, ultimately contributing to improved safety and performance of machines.
    \end{abstract}
    

\end{titlepage}
    
    \maketitle
\thispagestyle{firstpage}

    \section{Introduction} \label{sec:introduction}

Fault diagnosis systems play an essential role in engineered systems, which are continuously growing in complexity. Examples include high-precision production equipment \cite{Classens2021DigitalTwin,classens2023Opportunities} and aerospace applications \cite{marzatModelbasedFaultDiagnosis2012,zolghadriFaultDiagnosisFaultTolerant2014}, which require a large focus on improving system safety and reliability. Without proper monitoring and maintenance, the question is not whether a machine will fail but when it will. Therefore, timely detection and identification of faults is crucial to mitigate the risk of performance degradation, damage, and threats to human safety. In addition, the knowledge gained from diagnostic systems can be used to optimize maintenance scheduling. 

Complex engineered systems often operate in closed-loop configurations, where feedback mechanisms are essential to achieve the desired performance. The available models of these systems are inherently uncertain due to factors such as limited estimation accuracy, simplifications and assumptions, or system variability, which can significantly impact fault diagnosis performance. Consequently, there is a need for advanced fault diagnosis techniques that account for the closed-loop dynamics and can effectively manage these model uncertainties and the inevitable presence of disturbances in real-life settings.

Driven by the growing demand for more safe and reliable systems, the development of fault diagnosis approaches has received considerable attention. Both data-driven methods \cite{gertlerFaultDetectionDiagnosis1998,isermannFaultdiagnosisSystemsIntroduction2006,dingDatadrivenDesignFault2014} and model-based methods \cite{gertlerFaultDetectionDiagnosis1998,chenRobustModelBasedFault1999,isermannFaultdiagnosisSystemsIntroduction2006,dingModelbasedFaultDiagnosis2008,Varga2017} have shown considerable progress. Among these, observer-based methods as considered in this article have gained substantial attention due to their effectiveness in detecting various types of faults \cite{frank1997survey,chenRobustModelBasedFault1999}. Based on such an observer-based framework, strategies have been developed for fault detection and isolation (FDI) \cite{chenRobustModelBasedFault1999,dingModelbasedFaultDiagnosis2008,Varga2017}.

A key challenge in fault detection (FD) involves distinguishing faults from unknown disturbances. It is widely acknowledged that achieving satisfactory performance in model-based fault diagnosis systems requires a delicate balance between sensitivity to faults and disturbance rejection \cite{Ding2000}. Several optimal fault diagnosis methods have been developed for linear time-invariant (LTI) systems, including factorization-based techniques \cite{Ding2000,jaimoukhaMatrixFactorizationSolution2006}, often implemented through the solution of a Riccati equation \cite{liuOptimalSolutionsMultiobjective2007}. Alternatively, $\mathcal{H}_-/\mathcal{H}_\infty$ techniques have been proposed, employing LMI synthesis techniques \cite{houLMIApproachHinfty1996,liuLMIApproachMinimum2005,fengtaoFaultDetectionObserver2005,wangLMIApproachIndex2007}. These methods are optimal in the sense that the residual is as sensitive to faults as possible provided that the disturbance and plant model are exactly known.

An additional challenge in model-based fault diagnosis lies in distinguishing model errors from disturbances and faults. Robust methods have been developed to explicitly address modeling uncertainty. Among these, $\mathcal{H}_\infty$-norm bounded designs are often considered, which are optimized using $\mu$-synthesis \cite{sadrniaRobustInftyMu1997,stoustrupFaultEstimationStandard2002}. Other methods involve $\mathcal{H}_\infty$ model-matching techniques, solved through Linear Matrix Inequality (LMI) optimization \cite{zhongLMIApproachDesign2003}. Additionally, $\mathcal{H}_{-}/\mathcal{H}_\infty$ criteria are employed and addressed with LMI solutions \cite{haibowangIterativeLMIApproach2003} and non-smooth optimization \cite{henryTheoriesDesignAnalysis2021}, sometimes in combination with $\mu_{g}$-analysis \cite{newlinGeneralizationStructuredSingular1998}. Alternatively, the open-loop multiobjective problem is addressed by bounding model uncertainties \cite{Li2008}, in the context of LPV systems \cite{wei2025Application,huang2026towards}, or by IQC-based LPV formulations \cite{venkataraman2016robust,ho2019robust}.

A key challenge regarding robust methods involves its complexity and conservatism. For example, methods that rely solely on an $\mathcal{H}_{\infty}$ criterion do not directly account for $\mathcal{H}_{-}$ fault sensitivity, which has to be analyzed a posteriori. Model matching techniques often fail to guarantee optimality in terms of $\mathcal{H}_{-}$ fault sensitivity, as their effectiveness heavily depends on the reference model. In addition, these methods are often difficult to extend for fault isolation purposes. In general, many fault diagnosis methods are developed for open-loop systems and, therefore, not tailored to closed-loop systems, where the controller introduces directional closed-loop dynamics that can mask or amplify fault-induced residuals depending on the propagation of uncertainties through the feedback loop. 

Although several important steps have been taken towards fault diagnosis for complex systems, at present optimal detection of faults in uncertain closed-loop systems in an effective way remains challenging. This article builds on the preliminary work reported in \cite[Chpt. 2]{classensFaultDiagnosisUncertain2024} and aims to develop an optimal $\mathcal{H}_i/\mathcal{H}_\infty$ solution to the fault detection filter design problem for LTI multi-input multi-output (MIMO) uncertain closed-loop systems. The scope is confined to the fault detection task, however, the results can be directly extended for isolation purposes. The solution, solved using a single Riccati equation, is based on an upper-bound realization of the uncertainty and disturbance model and achieves an optimal compromise between the rejection of disturbances and modeling uncertainty with respect to fault sensitivity.

The upper-bound model is derived from worst-case gain analysis for systems with mixed uncertainties, encompassing both dynamic and parametric uncertainties, a known NP-hard problem \cite{braatzComputationalComplexitySpl1994}. To address this, lower and upper bounds are computed via skewed-$\mu$ power iterations \cite{hollandDevelopmentSkewLower2005,roosSystemsModelingAnalysis2013,balasRobustControlToolbox2024} and convex optimization employing D-G scaling \cite{hollandDevelopmentSkewUpper2005,packardComplexStructuredSingular1993,roosSystemsModelingAnalysis2013,balasRobustControlToolbox2024}. 
Skewed-$\mu$ power iterations use a heuristic to identify the parameter or complex matrix value corresponding to the worst-case lower bound at a given frequency. Subsequently, a stable LTI realization is constructed via interpolation \cite{Zhou1996a}. This concept was extended to construct worst-case mixed uncertainty samples that maximize gain across multiple frequencies \cite{patarticsWorstCaseUncertainty2023,patarticsConstructionUncertaintyMaximize2020}, utilizing nonlinear optimization and boundary Nevanlinna-Pick (BNP) interpolation \cite{ballInterpolationRationalMatrix1990}. This yields a stable, norm-bounded LTI uncertainty realization that interpolates a collection of matrix samples, providing a worst-case upper bound for the uncertainty and disturbance model, which is used in the proposed fault detection approach.

In summary, the key contributions are outlined as follows.
\begin{enumerate}[label=C\arabic*]
\label{chpt2:contribution1}
	\item An optimal solution is provided to the $\mathcal{H}_i/\mathcal{H}_\infty$ fault detection problem for closed-loop MIMO LTI uncertain systems. To this end, the uncertainty is extracted and a worst-case upper bound is constructed, which allows to solve the problem via factorization by means of a Riccati equation.
\label{chpt2:contribution2}
    \item The efficacy of the proposed approach is shown using a simulation model of an experimental reticle stage used in the lithography industry.
\end{enumerate}

The subsequent sections of this article are organized as follows: After the preliminaries, the fault detection filter optimization problem is formulated for closed-loop uncertain LTI systems in Section \ref{sec:problem formulation}. Next, in Section \ref{chpt2:sec:Solution}, a solution is proposed that optimally addresses the filter optimization problem, and the upper bound is further examined in Section \ref{sec:validation}. In Section \ref{sec:Example}, a numerical example is presented, demonstrating the application of the proposed solution on an experimental reticle stage. Finally, this article concludes by summarizing key findings.

\section{Notation and Preliminaries}
The sets of real numbers and nonnegative real numbers are indicated by $\mathbb{R}$ and $\mathbb{R}_{\geq 0}$. By $\norm{\cdot}_{2}$ the Euclidean norm is defined. A vector $d$ is a unitary vector if $\|{d}\|_2 = 1$. The maximum and minimum singular values of the matrix $A$ are denoted by $\bar{\sigma}(A)$ and $\underline{\sigma}(A)$, respectively. The real-rational subspace of $\mathcal{H}_{\infty}$ is denoted by $\mathcal{RH}_{\infty}$. The signal $y \in \mathcal{L}_{2}$ if $\norm{y}_{2}^{2} = \int_{0}^{\infty} y^{\top}(t) y(t) \mathrm{d} t < \infty$. 
A transfer function $N$ is called \textit{inner} if $N \in \mathcal{R}\mathcal{H}_{\infty}$ and $N^{H}N = I$ and co-inner if $N \in \mathcal{R}\mathcal{H}_{\infty}$ and $N N^{H} = I$. A transfer function $M$ is called \textit{outer} if $M \in \mathcal{R}\mathcal{H}_{\infty}$ and has full row normal rank and has no open right half plane zeros.

\begin{definition}(Linear fractional transformation) \label{definition:LFT}
For matrices $N$ and $M = \left[\begin{smallmatrix} M_{11} & M_{12} \\ M_{21} & M_{22} \end{smallmatrix} \right]$ of appropriate partitioning, the lower linear fractional transformation (LFT) is defined as $\mathcal{F}_{l}(M,N) = M_{11} + M_{12} N (I - M_{22} N )^{-1} M_{21}$ and the upper LFT as $\mathcal{F}_{u}(M,N) = M_{22} + M_{21} N (I - M_{11} N )^{-1} M_{12}$, under the assumption that the involved matrix inverses exists.
\end{definition}

Consider uncertainties whose values admit the structure
$\mathbf{\Delta}_{s} = \{ \mathrm{diag} ( p_{1} I, \ldots, p_{n_{r}} I, \delta_{1} I,$ $ \ldots, \delta_{n_{c}} I, \Delta_{1}, \ldots, \Delta_{n_{z}} ) \in \mathbb{C}^{p \times q}$ and whose blocks satisfy $p_j \in \mathbb{R}$ with $\abs{p_{j}} \leq 1$ for $j = 1,\ldots, n_{r}$, $\delta_j \in \mathbb{C}$ with $\abs{\delta_{j}} \leq 1$ for $j = 1,\ldots, n_{c}$, and $\Delta_j \in \mathbb{C}^{p_{j} \times q_{j}}$ with $\norm{\Delta_{j}}_{2} \leq 1$ for $j = 1,\ldots, n_{z}$. The actual set of uncertainties $\mathbf{\Delta} := \{ \Delta(s) \in \mathcal{RH}_\infty \; | \; \Delta(i \omega) \in \mathbf{\Delta}_{s} \; \mathrm{for} \; \mathrm{all} \; \omega \in \mathbb{R} \cup \{ \infty \} \}$. 

The structured singular value of a matrix $P$ with respect to the set $\mathbf{\Delta}$ is defined as
\begin{equation}
    \mu_{\mathbf{\Delta}} (P) = \frac{1}{\sup \{ r \; | \; \det(I - P \Delta) \neq 0 \; \mathrm{for} \; \mathrm{all} \; \Delta \in r \mathbf{\Delta} \}}.
\end{equation}

The following definitions are used throughout \cite{dingModelbasedFaultDiagnosis2008,liuOptimalSolutionsMultiobjective2007,Zhou1996a,skogestadMultivariableFeedbackControl2005}.

\begin{definition} \label{chpt2:lem:mingain}
(Minimum gain) The smallest gain of the stable continuous-time LTI system $G: \mathcal{L}_{2} \rightarrow \mathcal{L}_{2}$, named the $\mathcal{H}_{-}$ index, is defined as
\begin{equation}
\norm{G}_{-}^{[\infty]} = \inf_{\omega \in \mathbb{R}_{\geq 0}} \underline{\sigma}(G(j \omega)).
\end{equation}
\end{definition}
The $\mathcal{H}_{-}$ index over a finite frequency range is defined as $\norm{G}_{-}^{[\omega_1,\omega_2]} = \inf_{\omega \in [\omega_1,\omega_2]} \underline{\sigma}(G(j \omega))$ and the $\mathcal{H}_{-}$ index at zero frequency is defined as $\norm{G}_{-}^{[0]} = \underline{\sigma}(G(0))$. The minimum gain is not a norm and therefore named the $\mathcal{H}_{-}$ index. When no superscript is specified, $\norm{G}_{-}$ represents all possible definitions.

\begin{definition}
(Maximum gain)
The $\mathcal{H}_{\infty}$ norm of the continuous-time LTI system $G: \mathcal{L}_{2} \rightarrow \mathcal{L}_{2}$, denoted as $\norm{G}_{\infty}$, is given by
\begin{equation}
\norm{G}_{\infty} = \sup_{\omega \in \mathbb{R}_{\geq 0}} \bar{\sigma} (G(j \omega)).
\end{equation}
\end{definition}

\begin{lemma}(Robust Stability) \label{chpt2:lemma:RS}
Consider the system $N = {\left[\begin{array}{c|cc}
N_{11} \; & \; N_{12} \\
\hline N_{21} \; & \; N_{22} 
\end{array}\right]}$ in the $N\Delta$-structure, i.e., such that $\begin{bmatrix}
    u_{\Delta} \\ z
\end{bmatrix} = N \begin{bmatrix}
    y_{\Delta} \\ w
\end{bmatrix}$ and $y_{\Delta} = \Delta u_{\Delta}$ with $\Delta \in \mathbf{\Delta}$. Assume that the nominal system $N_{22}$ and the perturbations $\Delta$ are stable. Then, $\mathcal{F}_{u}(N,\Delta)$ is stable $\forall \Delta \in \mathbf{\Delta}$, if and only if $\:\: \Leftrightarrow \:\: \mu_{\Delta} (N_{11}) < 1, \forall \omega.$
\end{lemma}
    The proof is provided in \cite[Chapter 8]{skogestadMultivariableFeedbackControl2005}.

\vspace{5mm}
\section{{$\mathcal{H}_i/\mathcal{H}_\infty$}-problem formulation}
\label{sec:problem formulation}

%
\begin{figure}[bt]
    \centering
    \includegraphics[width = 200pt]{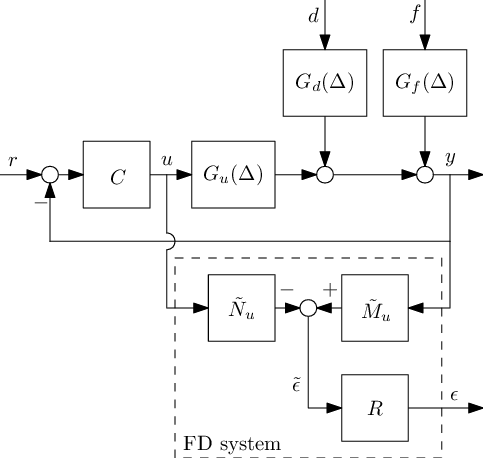}
    \caption{Generic fault detection configuration for uncertain closed-loop controlled systems. The control input $u$ and output $y$ form the inputs for the fault detection (FD) system which generates the residual signal $\epsilon$.}
    \label{chpt2:fig:FD_configuration_uncertain}
\end{figure}
%

Consider the input-output representation of the uncertain LTI process in the Laplace domain, described by
\begin{equation}\label{chpt2:uncertain_y}
    y=G_u(\Delta) u+G_d(\Delta) d+G_f(\Delta) f,
\end{equation}
where $G_u(\Delta)$, $G_d(\Delta)$ and $G_f(\Delta)$ are uncertain transfer function matrices (TFMs) from the control input $u$, disturbance $d$, and fault $f$, to the output $y$. In the time domain, the output at time $t \in \mathbb{R}_{\geq 0} = [0,\infty)$ is given by $y(t) \in \mathbb{R}^{n_y}$, the control input by $u(t) \in \mathbb{R}^{n_u}$, the disturbance by $d(t) \in \mathbb{R}^{n_d}$, and the fault by $f(t) \in \mathbb{R}^{n_f}$. The modelling uncertainty is denoted by $\Delta \in \boldsymbol\Delta$ and can be parametric or dynamic with suitable dimensions.
The system \eqref{chpt2:uncertain_y} is controlled to follow a reference $r$, with value $r(t) \in \mathbb{R}^{n_y}$ at time $t \in \mathbb{R}_{\geq 0} = [0,\infty)$, by means of a robustly stabilizing feedback controller $C$, i.e., $u = C(r-y)$, see Figure \ref{chpt2:fig:FD_configuration_uncertain}. Substitution of the feedback relation into \eqref{chpt2:uncertain_y} results in the closed-loop input-output relation for linear uncertain systems, given by
\begin{equation}
\begin{split}
     y&= S_\Delta\left(G_u(\Delta) C r+G_d(\Delta) d+G_f(\Delta) f\right),
\end{split}\label{chpt2:uncertain_CL_output}
\end{equation}
where $S_\Delta = (I + G_u(\Delta) C)^{-1}$ is the uncertain sensitivity function.

The closed-loop system is augmented with a fault detection system, which takes as inputs the control input $u$ and the output $y$. The fault detection system generates residual signals $\epsilon$, with value $\epsilon(t) \in \mathbb{R}^{n_y}$ at time $t \in \mathbb{R}_{\geq 0} = [0,\infty)$, that allow to detect faults $f$, despite the influence of the external disturbances $r$ and $d$. All residual generators can be parameterized as \cite{dingFaultDetectionFactorization1990}
\begin{equation}\label{chpt2:general_residual_gen}
    \epsilon = R
        \left(\tilde{M}_u y -\tilde{N}_u u\right),
\end{equation}
where $R \in \mathcal{R}\mathcal{H}_\infty^{n_y \times n_y}$ is a post-filter of the pre-residual $\tilde{\epsilon} = \tilde{M}_u y -\tilde{N}_u u$, which in the time-domain takes values $\tilde{\epsilon}(t) \in \mathbb{R}^{n_y}$. The transfer function matrices $\Tilde{M}_u, \Tilde{N}_u \in$ \RHinf form a left coprime factorization of the nominal plant $G_u(0)$, i.e., $G_u(0) = \Tilde{M}_u^{-1}\Tilde{N}_u$.
The uncertain feedback system \eqref{chpt2:uncertain_CL_output} augmented with the fault detection system \eqref{chpt2:general_residual_gen} is graphically depicted in Figure \ref{chpt2:fig:FD_configuration_uncertain}.

Let the uncertain part of the plant $G_u(\Delta)$ be defined as
\begin{equation}\label{chpt2:uncertain_Gu}
\tilde{G}_u(\Delta):=G_u(\Delta)-G_u(0).
\end{equation}
Substitution of \eqref{chpt2:uncertain_CL_output}, the control law $u=C(r-y)$, and \eqref{chpt2:uncertain_Gu} into \eqref{chpt2:general_residual_gen}, gives the residual dynamics for uncertain closed-loop systems as
\begin{align}\label{chpt2:residual_expression_long}
    \epsilon &= R\Tilde{M}_u\{\underbrace{\tilde{G}_u(\Delta) C S_\Delta}_{\textstyle T_{\tilde{\epsilon} r}^\Delta}r + \underbrace{(G_d(\Delta) - \tilde{G}_u(\Delta)CS_\Delta G_d(\Delta)}_{\textstyle T_{\tilde{\epsilon} d}^\Delta})d + \underbrace{(G_f(\Delta) - \tilde{G}_u(\Delta)CS_\Delta G_f(\Delta)}_{\textstyle T_{\tilde{\epsilon} f}^\Delta})f \}.
\end{align}
To enhance readability, let $T_{\tilde{\epsilon} r}(\Delta) = \tilde{M}_u T_{\tilde{\epsilon} r}^\Delta$, $T_{\tilde{\epsilon} d}(\Delta) = \tilde{M}_u T_{\tilde{\epsilon} d}^\Delta$ and $T_{\tilde{\epsilon} f}(\Delta) = \tilde{M}_u T_{\tilde{\epsilon} f}^\Delta$ be defined as the uncertain transfers from the reference $r$, the disturbances $d$, and the faults $f$, to the pre-residual $\tilde{\epsilon}$. 
Then, the residual dynamics for uncertain closed-loop systems can be written as
\begin{equation}\label{chpt2:epsilon_to_rdf}
    \epsilon = R 
    \underbrace{
    \begin{bmatrix}
        T_{\tilde{\epsilon} r}(\Delta) & T_{\tilde{\epsilon} d}(\Delta)
    \end{bmatrix}}_{\tilde{G}_d(\Delta)}
    \underbrace{\begin{bmatrix}
        r\\
        d
    \end{bmatrix}}_{\tilde{d}} +
    R T_{\tilde{\epsilon} f}(\Delta) f,
\end{equation}
where $\tilde{G}_d(\Delta)$ describes the transfer from the extended disturbance input $\tilde{d} = \begin{bmatrix}
    r & d
\end{bmatrix}^\top$ to the pre-residual $\tilde{\epsilon}$.

The residual dynamics clearly show the inherent trade-off between sensitivity to faults and robustness against external disturbances and the effect of modelling uncertainties. It is desirable to minimize $R\tilde{G}_d(\Delta)$ for all $\Delta \in \mathbf{\Delta}$ in order to reduce the impact of $r$ and $d$ on the residual, while, at the same time, it is also desirable to maximize $RT{\tilde{\epsilon} f}(\Delta)$ for all $\Delta \in \mathbf{\Delta}$ to enhance fault sensitivity.

\begin{remark}
    The residual dynamics in \eqref{chpt2:residual_expression_long} are equal to
    \begin{equation*}
    \begin{split}
    \epsilon &= R\tilde{M}_u S^{-1}S_\Delta\Big( \tilde{G}_u(\Delta)CS r + G_d(\Delta) d + G_f(\Delta) f \Big),
    \end{split}
    \end{equation*}
    see \ref{chpt2:Appendix:optimal_ratio}, where $S = (I+G(0)C)^{-1}$ is the nominal sensitivity. This expression is insightful since $S^{-1} S_\Delta$ can be factored out for $r$, $d$, and $f$.
\end{remark}

A natural way to evaluate the robustness against disturbances $\tilde{d}$ is through the $\mathcal{H}_\infty$-norm, whereas characterizing the sensitivity of faults necessitates a more intricate approach. The singular values of a matrix give a measure for the amplification in the direction of maximum action among all directions orthogonal to the singular vectors of any larger singular value. In essence, these give a measure for amplification in the principal directions of a system. In this context, all singular values $\sigma_i\left( R(j\omega) T_{\tilde{\epsilon} f}(j\omega,\Delta) \right), \omega \in [0,\infty)$, $\Delta \in \mathbf{\Delta}$, where $i = 1,\ldots, n_\sigma$ and $n_{\sigma} = \min(n_{y},n_{f})$, together form a measure of the fault sensitivity and the corresponding singular values cover all directions of the subspace spanned by $R(j\omega)T_{\tilde{\epsilon} f}(j\omega,\Delta)$. 

Given the measure for worst-case disturbance amplification and the measure for fault sensitivity, three different robust $\mathcal{H}_i/\mathcal{H}_\infty$ performance indices are defined. Each index is characterized by a progressively stricter criterion for fault sensitivity in its numerator. First, consider the index based on the worst-case fault sensitivity as
\begin{equation}\label{chpt2:performance_index_minmin}
    J_{\omega}(R) = \frac{\inf_{\Delta \in \mathbf{\Delta}} \left( \underline{\sigma} \left( R(j\omega) T_{\tilde{\epsilon} f}(j\omega,\Delta) \right) \right) }{\|R \tilde{G}_d(\Delta)\|_\infty}.
\end{equation}
Focusing on the worst-case in all principal directions gives $n_{\sigma}$ functions per frequency $\omega$ as
\begin{equation}\label{chpt2:performance_index}
    J_{i,\omega}(R) = \frac{\inf_{\Delta \in \mathbf{\Delta}} \left( \sigma_i\left( R(j\omega) T_{\tilde{\epsilon} f}(j\omega,\Delta) \right) \right) }{\|R \tilde{G}_d(\Delta)\|_\infty}.
\end{equation}
Considering the entire set $\bar{\Delta} \in \mathbf{\Delta}$ instead of just the worst-case in each principle direction, gives the criterion
\begin{equation}\label{chpt2:performance_index_wDelta}
    J_{i,\omega,\bar{\Delta}}(R) = \frac{ \sigma_i\left( R(j\omega) T_{\tilde{\epsilon} f}(j\omega,\bar{\Delta}) \right) }{\|R \tilde{G}_d(\Delta)\|_\infty}.
\end{equation}
Here, $\|R \tilde{G}_d(\Delta)\|_\infty$ for all $\Delta \in \mathbf{\Delta}$, whereas the fault sensitivity is based on a single $\bar{\Delta} \in \mathbf{\Delta}$. Both stem from the same set.

The objective is to find the fault detection filter that maximizes the ratio $J_{i,\omega,\bar{\Delta}}$ in \eqref{chpt2:performance_index_wDelta} for all singular values $i = 1,\ldots n_\sigma$, at every frequency $\omega$, and for every realization $\bar{\Delta} \in \mathbf{\Delta}$.

\begin{problem}\label{chpt2:prob:problem_Ding}
    Consider the residual dynamics \eqref{chpt2:residual_expression_long} of a closed-loop uncertain system described by \eqref{chpt2:uncertain_CL_output} and let $\gamma > 0$ be a user-defined combined disturbance and uncertainty rejection level such that $\|R \tilde{G}_d(\Delta)\|_\infty \leq \gamma$. Determine $R \in \mathcal{R}\mathcal{H}_\infty^{n_y \times n_y}$ such that $J_{i,\omega,\bar{\Delta}}(R)$ is maximized, i.e.,
\begin{equation}\label{chpt2:optim_problem}
\begin{split}
    \sup_{\scalebox{.75}{$R \in \mathcal{R}\mathcal{H}_\infty$}} \big\{
    \sigma_i\left( R(j\omega) T_{\tilde{\epsilon} f}(j\omega,\bar{\Delta}) \right) \Big| \|R \tilde{G}_d(\Delta)\|_\infty \leq \gamma
    \big\},
\end{split}
\end{equation}
for all $i = 1,\ldots, n_\sigma$, for all $\bar{\Delta} \in \mathbf{\Delta}$, and for all $\omega \in [0,\infty)$, where $\Delta \in \mathbf{\Delta}$. 
\end{problem}

\begin{remark}
    Note that the introduction of $\gamma>0$ has no effect on the performance indices in \eqref{chpt2:performance_index_minmin} to \eqref{chpt2:performance_index_wDelta}. Since the filter $R$ can be scaled arbitrarily, the bound $\gamma$ merely serves as a scaling parameter that has no influence on the optimal ratio, but ensures that the solution to Problem \ref{chpt2:prob:problem_Ding} is unique.
\end{remark}
\begin{remark}
    Note that Problem \ref{chpt2:prob:problem_Ding} is multiobjective since the solution is a single $R$ that solves the problem for all $i = 1, \ldots, n_\sigma$. Hence, the optimization problem includes the special $\mathcal{H}_{\infty}/\mathcal{H}_{\infty}$ and the $\mathcal{H}_{-}/\mathcal{H}_{\infty}$ objectives
\begin{equation}\label{chpt2:optim_problem2}
    \sup_{\scalebox{.75}{$R\in \mathcal{R}\mathcal{H}_\infty$}} \left\{
    \| R T_{\tilde{\epsilon} f}(\Delta) \|_\infty \Big| \|R \tilde{G}_d(\Delta)\|_\infty \leq \gamma
    \right\}, 
\end{equation}
    for all $\Delta \in \mathbf{\Delta}$ and
\begin{equation}\label{chpt2:optim_problem3}
    \sup_{\scalebox{.75}{$R\in \mathcal{R}\mathcal{H}_\infty$}} \left\{
    \| R T_{\tilde{\epsilon} f}(\Delta) \|_- \Big| \|R \tilde{G}_d(\Delta)\|_\infty \leq \gamma
    \right\}, 
\end{equation}
    for all $\Delta \in \mathbf{\Delta}$, respectively.
\end{remark}

\begin{remark}
(Special case: open loop) Although the framework presented in this study focuses on uncertain closed-loop systems, the same methodology and solution are equally applicable to open-loop uncertain systems. This requires deriving the residual dynamics in \eqref{chpt2:epsilon_to_rdf} for uncertain open-loop systems, yielding a residual $\epsilon$ that depends on the input $u$ rather than the reference $r$.
\end{remark}
\begin{remark} (Special case: nominal systems)
Note that when there is no uncertainty, i.e., $\Delta = 0$ and thus $T_{\tilde{\epsilon} r}^\Delta = 0, T_{\tilde{\epsilon} d}^\Delta = 0$, and $T_{\tilde{\epsilon} f}^\Delta = 0$, the residual dynamics in \eqref{chpt2:epsilon_to_rdf} simplify to the same expression for closed (and open)-loop nominal systems, that is
    \begin{equation*}
        \epsilon = R\tilde{M}_u(G_dd + G_ff).
    \end{equation*}
    Hence, the residual only depends on disturbances $d$ and possible faults $f$, and is independent of the reference $r$ \cite{Classens2021}. In this case, the FD filter design problem is equivalent for open and closed-loop nominal systems.
\end{remark}

\section{Robust closed-loop fault detection solution} \label{chpt2:sec:Solution}
In this section, the solution to Problem \ref{chpt2:prob:problem_Ding} is presented, which is contribution \ref{chpt2:contribution1}. First, an upper-bound envelope is established to encapsulate the worst-case scenario of the uncertain exogenous disturbances. Subsequently, this envelope is employed to solve the optimal robust fault detection filter design problem using the methodology originally developed for the uncertainty-free case in \cite{liuOptimalSolutionsMultiobjective2007}. This is followed by several remarks.

\begin{lemma} \label{chpt2:definition:untight}
Let $\bar{G}_d \in \mathcal{R}\mathcal{H}_\infty^{n_{y} \times (n_{y} + n_{d})}$ be an upper-bound envelope of $\tilde{G}_d(\Delta)$ which for all unitary vectors $\tilde{d}_2$ satisfies
\begin{equation}\label{chpt2:upperbound}
    \|\tilde{G}_d(j\omega,\Delta)\tilde{d}_1\|_2 \leq \|\bar{G}_d(j\omega)\tilde{d}_2\|_2 \quad \forall \omega, \forall \Delta \in \mathbf{\Delta},
\end{equation}
where $\tilde{d}_1$ is a corresponding unitary vector such that 
$\tilde{G}_d(j\omega,\Delta) d_{1} (j \omega) = g \bar{G}_d(j\omega) d_{2} (j \omega)$, with $g \in \mathbb{R}_{\geq 0}$. Then, condition \eqref{chpt2:upperbound} is satisfied if $\bar{\sigma}( \bar{G}_{do}^{-1}(j\omega) \tilde{G}_d(j\omega,\Delta) )$ $\leq 1$ for all $\omega$ and all $\Delta \in \mathbf{\Delta}$, where $\bar{G}_{do}$ is the outer of $\bar{G}_{d}$.
\end{lemma}

The proof is provided in \ref{app:proof:definition:untight}.

\begin{lemma} \label{chpt2:definition:tight}
Let $\bar{G}_d \in \mathcal{R}\mathcal{H}_\infty^{n_{y} \times (n_{y} + n_{d})}$ be a tight upper-bound envelope of $\tilde{G}_d(\Delta)$ which for all unitary vectors $\tilde{d}_2$ satisfies
\begin{equation} \label{chpt2:upperbound_2}
    \sup_{\Delta \in \mathbf{\Delta}} \|\tilde{G}_d(j\omega,\Delta)\tilde{d}_1\|_2 = \|\bar{G}_d(j\omega)\tilde{d}_2\|_2 \quad \forall \omega,
\end{equation}
where $\tilde{d}_1$ is a corresponding unitary vector such that 
$\tilde{G}_d(j\omega,\Delta) d_{1} (j \omega) = g \bar{G}_d(j\omega) d_{2} (j \omega)$, with $g \in \mathbb{R}_{\geq 0}$. Then, condition \eqref{chpt2:upperbound_2} is satisfied if and only if
$\sup_{\Delta \in \mathbf{\Delta}} \{ \sigma_i( \bar{G}_{do}^{-1}(j\omega) \tilde{G}_d(j\omega,\Delta) ) \} = 1$ for each $i = 1, \ldots, n_{\sigma}$, at every $\omega$, where $\bar{G}_{do}$ is the outer of $\bar{G}_{d}$.
\end{lemma}

The proof is provided in \ref{app:proof:definition:tight}.

Hence, a tight upper bound implies that the worst-case amplification of $\tilde{G}_d(j\omega,\Delta)$ is equal to the amplification of $\bar{G}_d(j\omega)$ at every $\omega$. Next, the complete solution to the robust fault detection filter design problem is presented.

\begin{theorem} \label{chpt2:theorem:main_solution}
  Let the optimal fault detection filter which solves Problem \ref{chpt2:prob:problem_Ding} be parameterized as 
        \begin{equation}
            \epsilon = R_{\mathrm{opt}} \label{chpt2:eq:eps_opt}
            \begin{bmatrix}
                \tilde{M}_u & -\tilde{N}_u
            \end{bmatrix}
            \begin{bmatrix}
                y\\
                u
            \end{bmatrix},
        \end{equation}
        and assume that the following assumptions are satisfied.
    \begin{enumerate}
    \setcounter{enumi}{0}
           \item $\bar{G}_d$ has state space representation $(A, B_d, C, D_d)$ with $A$ Hurwitz and $(A,C)$ detectable.\label{chpt2:assumption_1}
            \item $D_d$ has full row rank.\label{chpt2:assumption_2}
            \item $\bar{G}_d$ has no transmission zeros on the imaginary axis.\label{chpt2:assumption_3}
           \item $\bar{G}_d$ is a tight upper-bound realization of $\tilde{G}_d(\Delta)$. \label{chpt2:assumption_4}
        \end{enumerate}
        Then, there exists a co-inner-outer factorization of $\bar{G}_d = G_{do} G_{di}$ with the optimal post-filter $R_{\mathrm{opt}} = \gamma G_{do}^{-1}$ and $\Tilde{M}_u,\Tilde{N}_u \in$ \RHinf any LCF of the nominal system $G_u(0)$, which achieves for all $\omega$, $i=1,\ldots, n_\sigma$, and $\bar{\Delta} \in \mathbf{\Delta}$
        \begin{equation*}
            \sup_{R \in \mathcal{R}\mathcal{H}_\infty} J_{i,\omega,\bar{\Delta}}(R_{\mathrm{opt}}) = 
            \sigma_i\left( R(j\omega) T_{\tilde{\epsilon} f}(j\omega,\bar{\Delta}) \right).
        \end{equation*}          
        The corresponding state-space representation of $R_{\mathrm{opt}}$ is given by
        \begin{equation}
            R_{\mathrm{opt}} = \gamma
            \left[\begin{array}{c|c}
A+L_0 C & L_0 \\
\hline R_d^{-1 / 2} C & R_d^{-1 / 2}
\end{array}\right] \in \mathcal{R}\mathcal{H}_\infty,
        \end{equation}
        in which $R_d := D_d D_d^T > 0$ and $Y \geq 0$ is the stabilizing solution to the Riccati equation
        \begin{align*}
            \left( A-B_d D_d^T R_d^{-1} C \right)Y + 
            Y\left( A-B_d D_d^T R_d^{-1} C \right)^T - YC^TR_d^{-1}CY + 
            B_d\left( I-D_d^TR_d^{-1} D_d \right)B_d^T = 0,
        \end{align*}
        such that $A-B_d D_d^T R_d^{-1} C - Y C^T R_d^{-1} C$ is stable and 
        \begin{equation}
            L_0 := -\left( B_d D_d^T + YC^T \right)^T.
        \end{equation}
\end{theorem}

The proof is provided in \ref{app:proof:theorem:main_solution}.

Next, a series of observations is presented concerning assumptions \ref{chpt2:assumption_1} to \ref{chpt2:assumption_3}. Following this, an interpretable explanation of the optimal filter is given and an analysis is provided on the implications of the derived result with respect to Assumption \ref{chpt2:assumption_4}.
\begin{remark}
Although it may appear that requiring $A$ to be Hurwitz in Assumption \ref{chpt2:assumption_1} is restrictive, this is not the case. In fact, the system described by the TFM $\bar{G}_d$ is always stable, as it characterizes the closed-loop TFM $\tilde{d} \rightarrow \tilde{\epsilon}$.
\end{remark}

\begin{remark}
Assumption \ref{chpt2:assumption_2} implies that $n_y \leq n_d$ and that all outputs are subject to some form of disturbance or measurement noise. It can be argued that this assumption can be made without loss of generality, as it is practically impossible to obtain perfect measurements in any system. Additionally, it is reasonable to assume that the measurement noise is independent of each other. Hence, it is reasonable to assume that the disturbance matrix $D_{d}$ has full row rank. Readers are referred to \cite{liuOptimalSolutionsMultiobjective2007} for an adaptation of the state-space matrices of $G_{d}$ in case its $D$ matrix is not full row rank.
\end{remark}

\begin{remark}
In numerous applications, Assumption \ref{chpt2:assumption_3} is not considered restrictive, as many systems inherently include a certain level of damping, resulting in zeros positioned away from the imaginary axis. However, in cases where Assumption \ref{chpt2:assumption_3} is not satisfied, alternative approaches are available if $\bar{G}_d$ has zeros on the imaginary axis \cite{Glover2011}.
\end{remark}

%
\begin{figure}[t]
    \centering
    \includegraphics[width=200pt]{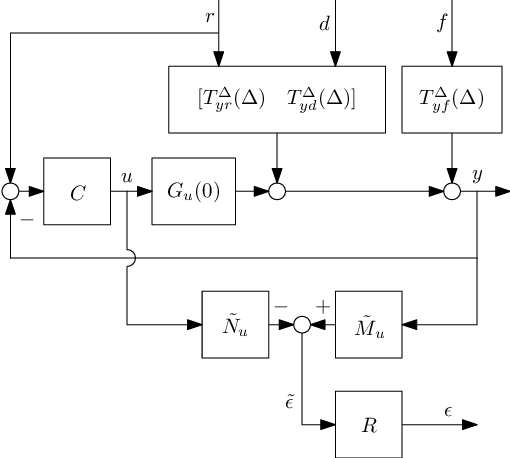}
    \caption{Generic fault detection configuration for uncertain closed-loop controlled systems. The control input $u$ and output $y$ form the inputs for the fault detection (FD) system which generates the residual signal $\epsilon$.}
    \label{chpt2:fig:FD_configuration_uncertain_CL_as_OL}
\end{figure}
\begin{remark}
    Note that the achieved performance $J_{i,\omega,\bar{\Delta}}(R)$ is independent of the choice of observer gain matrices $L$ in the LCF of the nominal plant $G(0)$ used in the filter design, see  \cite[Theorem 13.37]{Zhou1996a}.
\end{remark}

\begin{remark}
    Note that the solution is completely determined by the uncertainty and disturbance models encaptured in $\tilde{G}_{d}(\Delta)$ and is therefore independent of the fault model $\tilde{G}_{f}(\Delta)$.
\end{remark}

The presented strategy can be interpreted as follows. First, the closed-loop configuration in Figure \ref{chpt2:fig:FD_configuration_uncertain} is reconfigured to the form in Figure \ref{chpt2:fig:FD_configuration_uncertain_CL_as_OL}. Here, the uncertainty inherently present in the closed loop is located in the blocks $[T_{yr}^{\Delta}(\Delta) \quad T_{yd}^{\Delta}(\Delta)]$ and $T_{yf}^{\Delta}(\Delta)$. The LCF $\tilde{N}_u$, $\tilde{M}_u$ cancels the effect of the nominal closed loop and the remaining robust fault-to-disturbance optimization is performed by $R$. To be precise, this filter is based on the overbound $\bar{G}_{d}$ of $\tilde{M}_{u} [T_{yr}^{\Delta}(\Delta) \quad T_{yd}^{\Delta}(\Delta)]$.

To interpret the optimal post-filter, $R = \gamma G_{do}^{-1}$, recall that the co-outer of a TFM can be interpreted as a frequency domain magnitude profile. Hence, by multiplication with the inverse of the magnitude profile, the worst-case influence of the exogenous disturbance $\tilde{d}$ on the residuals $\epsilon$ is homogenized across the entire subspace. As a consequence, the optimal solution affects the TFM from faults $f$ to residuals $\epsilon$ by an inverse weighting of this magnitude profile.

To attain the supremum as defined in \eqref{chpt2:optim_problem} and thereby achieve optimal fault detection performance, it is imperative that the upper-bound realization satisfies Assumption \ref{chpt2:assumption_4}. This condition is inherently met if $\bar{G}_d \in \tilde{G}_d(\Delta)$ and \eqref{chpt2:upperbound} holds for all frequencies $\omega$ and uncertainties $\Delta\in\mathbf{\Delta}$. In such instances, the worst-case realization consistently bounds the other realizations across all frequencies from above. However, if the upper bound complies with \eqref{chpt2:upperbound} but is not tight, conservatism is introduced resulting in a suboptimal, yet still robust, filter in view of \eqref{chpt2:optim_problem}. In this case, the result is merely optimal given the conservative upper bound.

Next, the design of an optimal fault detection filter and an approach to minimize conservatism in the upper-bound envelope $\bar{G}_d$ is examined. Additionally, remarks are given regarding the uncertainty modeling. Subsequently, the optimal filter is computed and applied to a next-generation motion stage, and the time-domain responses for fault detection are shown.

\section{Verification and synthesis of the upper bound} \label{sec:validation}
The optimal filter in Theorem~\ref{chpt2:theorem:main_solution} relies on the existence of a tight upper-bound envelope $\bar{G}_d$, as required by Assumption~\ref{chpt2:assumption_4}. The practical realization of this assumption requires both the verification of candidate envelopes against Lemmas~\ref{chpt2:definition:untight} and \ref{chpt2:definition:tight}, and the construction of suitable envelopes in practice. This section addresses these aspects. First, verification conditions based on the structured singular value are presented. Subsequently, two methods are proposed for synthesizing candidate upper-bound envelopes.

\subsection{Verification of upper-bound envelopes}
Given a candidate upper-bound envelope $\bar{G}_d$, it must be verified whether the conditions of Lemma~\ref{chpt2:definition:untight} or Lemma~\ref{chpt2:definition:tight} are satisfied. This verification determines whether $\bar{G}_d$ is admissible for the robust fault detection filter design of Theorem~\ref{chpt2:theorem:main_solution}. The procedure is first presented for $n_y=1$ and subsequently generalized to arbitrary output dimension.

For systems with $n_y=1$, condition \eqref{chpt2:upperbound} can be verified by checking $\bar{\sigma}( \tilde{G}_d(j\omega,\Delta) ) \leq \bar{\sigma}( \bar{G}_d(j\omega))$ for all $\omega$
and $\Delta \in \mathbf{\Delta}$. More rigorously, it suffices to verify $\mu_{\mathbf{\Delta}} (\tilde{G}_d(j\omega,\Delta)) \leq \bar{\sigma}( \bar{G}_d(j\omega))$ for all $\omega$, where $\mu_{\mathbf{\Delta}}$ denotes the structured singular value associated with $\mathbf{\Delta}$. In case the inequality holds, $\bar{G}_d$ satisfies Lemma~\ref{chpt2:definition:untight}; if it is satisfied with equality, Lemma~\ref{chpt2:definition:tight} holds.

For systems with arbitrary $n_{y}$, condition \eqref{chpt2:upperbound} cannot, in general, be verified by directly comparing the singular values of $\bar{G}_d$ to $\tilde{G}_d(\Delta)$. Instead, it is required to first compute part of the solution, i.e., the co-outer factor of $\bar{G}_d$. Then, condition \eqref{chpt2:upperbound} is satisfied if $\bar{\sigma}( \bar{G}_{do}^{-1}(j\omega) \tilde{G}_d(j\omega,\Delta) ) \leq 1$ for all $\omega$ and all $\Delta \in \mathbf{\Delta}$. A comprehensive verification over all $\Delta \in \mathbf{\Delta}$ is obtained via the structured singular value. If $\mu_{\mathbf{\Delta}} (\bar{G}_{do}^{-1}(j\omega) \tilde{G}_d(j\omega,\Delta)) \leq 1$ holds for all $\omega$, with $\mu_{\mathbf{\Delta}}$ defined with respect to the structure $\mathbf{\Delta}$, then $\bar{G}_d$ satisfies Lemma~\ref{chpt2:definition:untight}. 

To satisfy Lemma~\ref{chpt2:definition:tight}, all singular values are relevant. Specifically, $\sup_{\Delta \in \mathbf{\Delta}} \{ \sigma_i( \bar{G}_{do}^{-1}(j\omega) \tilde{G}_d(j\omega,\Delta) ) \} = 1$ must hold for every $i = 1, \ldots, n_{\sigma}$, and all $\omega$. Consequently, the closer
$\sup_{\Delta \in \mathbf{\Delta}} \{\sigma_i( \bar{G}_{do}^{-1}(j\omega) \tilde{G}_d(j\omega,\Delta) ) \}$ is to unity from below for all $i$, the tighter the upper-bound envelope and the lower the associated conservatism. This interpretation is illustrated in the example of Section~\ref{chpt2:sec:FFR_robustFDFsyn}.

\subsection{Synthesis of upper-bound envelopes}
Having established verification conditions for candidate envelopes, the next step is to construct an upper-bound envelope $\bar{G}_d$ that satisfies Assumption~\ref{chpt2:assumption_4} or introduces as little conservatism as possible. This subsection presents two approaches for obtaining an initial realization $\bar{G}_{d,\mathrm{init}}$, which can subsequently be scaled to satisfy the conditions of Lemma~\ref{chpt2:definition:untight}. Consider the scaled $\bar{G}_{d,\mathrm{init}}$
\begin{equation}\label{chpt2:Gdbar_wc_alpha_beta}
    \bar{G}_d = 
    W_{o} \bar{G}_{d,\mathrm{init}} W_{i},
\end{equation}
where $W_o \in \mathcal{R}\mathcal{H}_\infty^{n_{y} \times n_{y}}$ and $W_i \in \mathcal{R}\mathcal{H}_\infty^{(n_{y} + n_{d}) \times (n_{y} + n_{d})}$ are transfer function matrices for scaling. Next, the two methods are described to obtain $\bar{G}_{d,\mathrm{init}}$. 
First, based on a single frequency that corresponds to the worst-case gain, and second, through interpolation over multiple frequencies using nonlinear optimization and boundary Nevanlinna-Pick interpolation.

\subsubsection{Upper bound based on worst-case frequency}
An initial $\bar{G}_{d,\mathrm{init}}(s)$ is derived from a worst-case gain analysis \cite{packardResultsWorstcasePerformance2000} as follows. Assume that this $\bar{G}_{d,\mathrm{init}}(\Delta)$ is parameterized by the LFT
\begin{equation*}
    \bar{G}_{d,\mathrm{init}}(\Delta) = \mathcal{F}_{u} \left( D (s), \tilde{\Delta}(s) \right),
\end{equation*}
with structured mixed uncertainty $\tilde{\Delta} \in \tilde{\mathbf{\Delta}}$ and $D(s)$ of compatible dimensions, where 
\begin{equation*}
    \tilde{\mathbf{\Delta}} : = \left\{ \mathrm{diag} (\tilde{\Delta}_{p},\tilde{\Delta}_{d}(s)), \tilde{\Delta}_{p} \in \tilde{\mathbf{\Delta}}_{p}, \tilde{\Delta}_{d} \in \tilde{\mathbf{\Delta}}_{d} \right\},
\end{equation*}
in which the subscripts $p$ and $d$ refer to parametric and dynamic respectively, and
\begin{equation*}
\begin{split}
    &\tilde{\mathbf{\Delta}}_{p} := \{ \mathrm{diag} ( \tilde{p}_{1} I_{n_{1}}, \ldots, \tilde{p}_{n_{r}} I_{n_{r}} ), \tilde{p}_{j} \in \mathbb{R}, \abs{\tilde{p}_{i}} \leq 1, j = 1, \ldots, n_{r} \},
\end{split}
\end{equation*}
\begin{equation*}
\begin{split}
    &\tilde{\mathbf{\Delta}}_{d} := \{ \mathrm{diag} ( \tilde{\Delta}_{1}(s), \ldots, \tilde{\Delta}_{n_{z}}(s) ), \norm{\tilde{\Delta}_{j}(s)}_{\infty} \leq 1, j = 1, \ldots, n_{z} \}.
\end{split}
\end{equation*}

\begin{definition}
    The worst-case gain of the uncertain system is
    \begin{equation}
        \bar{\Gamma} := \sup_{\tilde{\Delta}(s) \in \mathbf{\Delta}} \norm{ \mathcal{F}_{u} (D(s), \tilde{\Delta}(s)) }_{\infty}, \label{chpt2:eq:wcgainGam}
    \end{equation}
    if $\mathcal{F}_{u} (D(s), \tilde{\Delta}(s)))$ is robustly stable and $\bar{\Gamma} = \infty$ otherwise.
\end{definition}
This worst-case gain is equivalently calculated by computing the peak gain frequency by frequency and then taking the maximum. To formalize this, define the set of block structured complex uncertainties
\begin{equation*}
\begin{split}
    &\mathbf{Q}_{d} := \{ \mathrm{diag}(Q_{1}, \ldots, Q_{n_{Q}}), Q_{j} \in \mathbb{C}^{r_{j} \times c_{i}}, \bar{\sigma}(Q_{j}) = 1, \mathrm{rank}(Q_{j}) = 1, j = 1,\ldots, n_{Q} \}.
\end{split}
\end{equation*}
It is sufficient to restrict the complex matrices to be rank one \cite{packardComplexStructuredSingular1993}. The mixed set is defined as
\begin{equation*}
\begin{split}
    &\mathbf{Q} : = \left\{ \mathrm{diag} (\tilde{\Delta}_{p}, Q_{d}), \tilde{\Delta}_{p} \in \tilde{\mathbf{\Delta}}_{p}, Q_{d} \in \mathbf{Q}_{d} \right\},
\end{split}
\end{equation*}
Now consider the worst-case gain \eqref{chpt2:eq:wcgainGam}, equivalent to
    \begin{equation*}
        \bar{\Gamma} := \max_{\omega \in \mathbb{R} \cup \infty}\Gamma (\omega),
    \end{equation*}
where the peak gain at a specific $\omega$ is defined as
    \begin{equation*}
        \Gamma (\omega) := \sup_{Q \in \mathbf{Q}} \bar{\sigma}(\mathcal{F}_{u} (D(j \omega), Q)),
    \end{equation*}
and the frequency corresponding to the worst-case peak gain is
\begin{equation*}
    \omega_{wc} = \arg \max_{\omega \in \mathbb{R} \cup \infty}\Gamma (\omega).
\end{equation*}    
Lower bounds $L_{k}$ and upper bounds $U_{k}$ of $\Gamma (\omega)$ are computed for each frequency $\omega_{k}$ such that $L_{k} \leq \Gamma(\omega_{k}) \leq U_{k}$ \cite{hollandDevelopmentSkewLower2005,hollandDevelopmentSkewUpper2005}. Skewed-$\mu$ power iterations at a particular frequency $\omega_{k}$ then yield $Q_{k} \in \mathbf{Q}$ such that
\begin{equation*}
    \bar{\sigma} ( \mathcal{F}_{u} (D(j \omega_{k}), Q_{k} ) ) = L_{k}.
\end{equation*}
Hence, these power iterations allow to compute $\omega_{wc}$ and $Q_{k,wc}$ that correspond to $\bar{\Gamma}$. The parametric part found $\tilde{\Delta}_{p,wc}$ can be directly substituted in $\tilde{\Delta}_{wc} \in \tilde{\mathbf{\Delta}}$. The $\tilde{\Delta}_{d,wc}$ is found by construction of an LTI uncertainty that interpolates $Q_{d,wc}$ at $\omega_{wc}$ \cite[proof of Theorem 9.1]{Zhou1996a}. With this result, the worst-case gain realization at $\omega_{wc}$ is then constructed as $\bar{G}_{d,\mathrm{init}}(s,\Delta) = \mathcal{F}_{u} \left( D(s), \tilde{\Delta}_{wc}(s) \right)$, with $\tilde{\Delta}_{wc}(s) = \mathrm{diag} ( \tilde{\Delta}_{p,wc}, \tilde{\Delta}_{d,wc}(s))$.

The approach in this subsection yields a sample of the uncertain system that maximizes the gain at $\omega_{wc}$, the frequency where the peak gain occurs. This $\bar{G}_{d,\mathrm{init}}$ may have the largest gain over all frequencies, but is not necessarily high at other frequencies. 
In such case, $\bar{G}_{d,\mathrm{init}}$ can be tuned using $W_{o}$ and $W_{i}$ such that $\mu_{\mathbf{\Delta}} (\bar{G}_{do}^{-1}(j\omega) \tilde{G}_d(j\omega,\Delta)) \leq 1$. Next, an approach is presented that allows to obtain $\tilde{\Delta}_{d}$ by interpolation over multiple frequencies. 

\subsubsection{Upper bound based on multiple frequency interpolation}
Consider a collection of frequencies $\left\{ \omega_k \right\}_{k=1}^{n_{\omega}}$ with worst-case gain lower bounds $\left\{ L_k \right\}_{k=1}^{n_{\omega}}$, and cost $\tilde{J} : \mathbf{\Delta} \rightarrow \mathbb{R}$,
\begin{equation}
    \tilde{J}(\tilde{\Delta}(s)) := \sum_{k=1}^{n_{\omega}} \bar{\sigma} \left( \mathcal{F}_{u} \left( D(j \omega_k), \tilde{\Delta}(j \omega_k) \right) \right). \label{chpt2:eq:Jtilde}
\end{equation}
The construction problem over multiple frequencies is to find the structured $\tilde{\Delta}_{m}(s) \in \mathbf{\Delta}$, where $m$ refers to multiple, which maximizes $\tilde{J}(\tilde{\Delta}(s))$. 

The theoretical upper bound to $\tilde{J}(\tilde{\Delta}(s))$ is $\tilde{J}_{u} := \sum_{k=1}^{n_{\omega}} L_k$. If a system only has dynamic uncertainties, it is possible to find $\tilde{\Delta}_{m}(s)$ such that $\tilde{J}(\tilde{\Delta}_{m}(s)) = \tilde{J}_{u}$. When mixed uncertainties are present, the theoretical upper bound is generally unattainable because parametric uncertainties couple frequencies. Specifically, different values of the same parametric uncertainty might be needed to achieve the lower bound $L_{k}$ at each frequency.

To find the realization $\tilde{\Delta}_{m}(s) \in \tilde{\mathbf{\Delta}}$, nonlinear optimization is used to find $\tilde{\Delta}_{p,m} \in \tilde{\mathbf{\Delta}}_{p}$ and the complex matrices $\{ Q_{d,k} \}_{k=1}^{n_{\omega}}$. Subsequently, the interpolant $\tilde{\Delta}_{d,m}(s) \in \tilde{\mathbf{\Delta}}_{d}$ is computed using the boundary Nevanlinna Pick method \cite[Example 21.3.1 and Corollary 21.4.2]{ballInterpolationRationalMatrix1990}. The reader is referred to \cite[Algorithm 2]{patarticsWorstCaseUncertainty2023} for details regarding the nonlinear optimization using the interior-point algorithm and application of the boundary Nevanlinna Pick method.

\begin{remark}
    Be aware that the construction of the perturbed models $\bar{G}_{d,\mathrm{init}}$ focuses on the largest singular value. Still, $\bar{G}_{d}$ must be weighted using $W_{o}$ and $W_{i}$ such that $\mu_{\mathbf{\Delta}} (\bar{G}_{do}^{-1}(j\omega) \tilde{G}_d(j\omega,\Delta)) \leq 1$. This amounts to checking and possibly weighting the other principal directions as well.
\end{remark}

\begin{remark}
    A conservative upper-bound realization $\bar{G}_d$ can always be constructed based solely on the first singular value $\bar{\sigma}(\tilde{G}_d(\Delta))$. Here, a weight $W$ is created such that $\bar{\sigma}(\tilde{G}_d(\Delta)) \leq W$. Subsequently $\bar{G}_d$ is constructed as $\bar{G}_d = W I$ which results in a diagonal post-filter $R$ with equal singular values.
\end{remark}

\section{Example}
\label{sec:Example}
In this section, the proposed approach is validated using an uncertain model of an experimental reticle stage, see Figure \ref{chpt2:fig:FFR} and \ref{chpt2:fig:FFR_Closeup}. The reticle stage is actuated by 8 force actuators in the vertical direction and the position is measured with 8 capacitive sensors with nanometer accuracy, of which a two-input two-output uncertain equivalent plant is considered to illustrate the approach. Subsequently, a robust fault detection filter is designed in which the upper bound $\bar{G}_{d}$ is examined in detail, as well as its influence on the closed-loop transfer functions to the residuals. Finally, the time-domain response is computed and discussed.

\subsection{Uncertain system model and closed-loop configuration}
Consider the nominal model of order 14 which is expanded with an uncertain channel in Figure \ref{chpt2:fig:FFR_Gu}. The generalized plant $P$ is partitioned as
\begin{equation}
    \begin{bmatrix}
    \begin{array}{c}
         z_\Delta\\
         \hline
        \tilde{y}_{1}\\
        \tilde{y}_{2}
    \end{array}
    \end{bmatrix}
    =
    \underbrace{
    \begin{bmatrix}
    \begin{array}{c|cc}
         G_{z_\Delta w_\Delta} & G_{z_\Delta u_{1}} & G_{z_\Delta u_{2}}\\
         \hline
        G_{y_{1} w_\Delta} & G_{11} & G_{12}\\
        G_{y_{2} w_\Delta} & G_{21} & G_{22}
    \end{array}
    \end{bmatrix}}_P
    \begin{bmatrix}
    \begin{array}{c}
        w_\Delta\\
        \hline
        u_{1}\\
        u_{2} 
    \end{array}
    \end{bmatrix},
\end{equation}
with uncertainty channel $w_\Delta = \Delta z_\Delta$ and dynamic norm-bounded uncertainty $\Delta \in \mathbf{\Delta}$. The uncertain plant is constructed as $G_u = \mathcal{F}_u(P,\Delta)$. The disturbance model $G_{d}$ and the flat-spectrum fault model $G_{f}$ are shown in Figure \ref{chpt2:fig:FFR_GdGf}.

The system operates in closed loop, see Figure \ref{chpt2:fig:FD_configuration_uncertain}, and is robustly stabilized by controller with a target bandwidth of 140 Hz that is synthesized with a one-step robust control approach \cite{Tacx2022}.

\vspace{-2mm}
\begin{figure}[b]
\centering
    \begin{subfigure}[b]{0.45\textwidth}
\centering
\setlength{\fboxsep}{-1pt}%
\setlength{\fboxrule}{1pt}%
\fbox{\includegraphics[width=200pt]{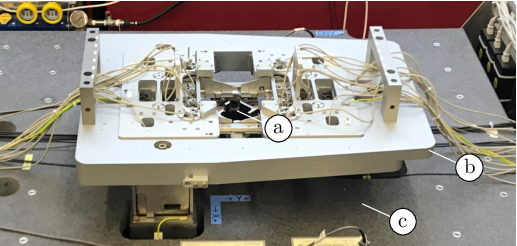}}
\caption{Overview, where \encircle{a} indicates the moving stage (the reticle) and \encircle{b} indicates the force frame. The vibration isolation system is denoted by \encircle{c}.}
\label{chpt2:fig:FFR}
    \end{subfigure}
    \hspace{5mm}
    \begin{subfigure}[b]{0.45\textwidth}
\centering
\setlength{\fboxsep}{-1pt}%
\setlength{\fboxrule}{1pt}%
\fbox{\includegraphics[width=200pt]{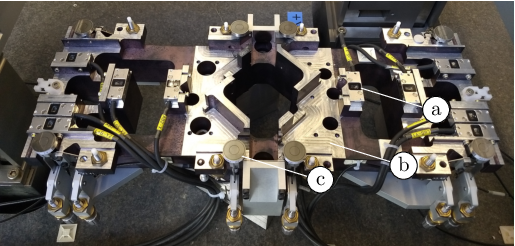}}
\caption{Close up, where \encircle{a} is a horizontal position sensor, \encircle{b} is the metrology frame, and \encircle{c} is a capacitive vertical position sensor.}
\label{chpt2:fig:FFR_Closeup}
    \end{subfigure}
    \caption{Overview and close up of the reticle stage.}
\end{figure}

\vspace{-3mm}
\begin{figure}[b!]
    \centering
    \includegraphics[scale=0.845]{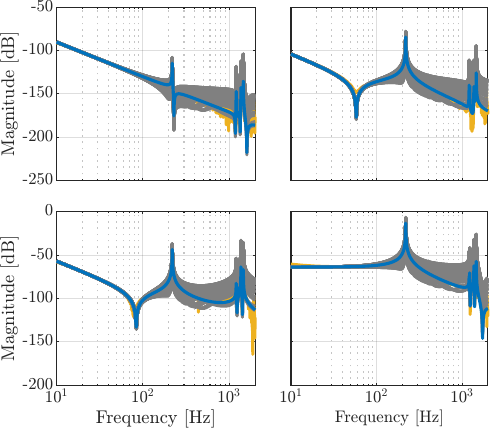}
    \caption{Bode magnitude plot of the reticle stage. The measured frequency response function is depicted in $\tikzline{MatlabYellow}$ together with the nominal 14th order plant model $G$ in $\tikzline{MatlabBlue}$ and the uncertain state space model $G_u$ in $\tikzline{MatlabGray50}$ visualized for 50 different uncertainties $\Delta \in \mathbf{\Delta}$.}
    \label{chpt2:fig:FFR_Gu}
\end{figure}

\newpage \subsection{Robust fault detection filter synthesis}\label{chpt2:sec:FFR_robustFDFsyn}
Next, the introduced models are used to synthesize a robust fault detection filter using the solution proposed in Section \ref{chpt2:sec:Solution}. First, a stable LCF, $\Tilde{M}_u$, $ \tilde{N}_u$, of the nominal system $G_u(0)$ is computed. Given these factors, the uncertain plant $G_u(\Delta)$, disturbance and fault models $G_{d}$ and $G_{f}$, respectively, and the robustly stabilizing feedback controller $C$, the closed-loop transfers from exogenous inputs to the pre-residual, $[
    r \:\: d \:\: f ] \rightarrow \Tilde{\epsilon}$, are computed, which allows to determine $\tilde{G}_d^{2\times4}(\Delta)$, see \eqref{chpt2:epsilon_to_rdf}. 

To design the upper-bound envelope $\Bar{G}_d$, frequencies $\left\{ \omega_k \right\}_{k=1}^{n_{\omega}}$ are selected: $\omega_{1} = 1$ Hz, $\omega_{2} = \omega_{wc} = 219$ Hz, $\omega_{3} = 1202$ Hz, and $\omega_{4} = 1368$ Hz. Maximizing $\tilde{J}$, see \eqref{chpt2:eq:Jtilde}, yields $\Delta(j \omega_{k})$, and interpolation via BNP gives $\Delta_{np}$. The perturbed model maximizing $\tilde{J}$ is $\bar{G}_{d,\mathrm{init}} = \mathcal{F}_{u} \left( D, \Delta_{np} \right)$. Figure \ref{chpt2:fig:FFR_svd_Gdbar} shows the singular values of $\tilde{G}_d(\Delta)$, $\mathcal{F}_{u} \left( D(s), \Delta_{wc} \right)$, and $\bar{G}_{d,\mathrm{init}}$.

The upper-bound $\bar{G}_{d}$ is iteratively updated via $W_o$ and $W_i$, see Figure \ref{chpt2:fig:FFR_Walpha_beta}, to satisfy Lemma \ref{chpt2:definition:untight}. The singular values of $\bar{G}_{d}(s)$ are shown in Figure \ref{chpt2:fig:FFR_svd_Gdbar}. Using \cite[Theorem 13.37]{Zhou1996a} the co-outer factor is computed, and the optimal post-filter $R$ is obtained via Theorem \ref{chpt2:theorem:main_solution} with $\gamma = 1$. The condition $\mu_{\mathbf{\Delta}} (\bar{G}_{do}^{-1}(j\omega) \tilde{G}_d(j\omega,\Delta)) \leq 1$, see Figure \ref{chpt2:fig:SSV}, and singular values  $\sigma_i( T_{\varepsilon \tilde{d}} ) \leq 1$, see Figure \ref{chpt2:fig:FFR_svd_transfers} guide local adjustments to $W_{o}$ and $W_{i}$. The final post-filter $R$ is shown in Figure \ref{chpt2:fig:FFR_Q}.

\begin{figure}[t]
    \centering
    \includegraphics[scale=0.845]{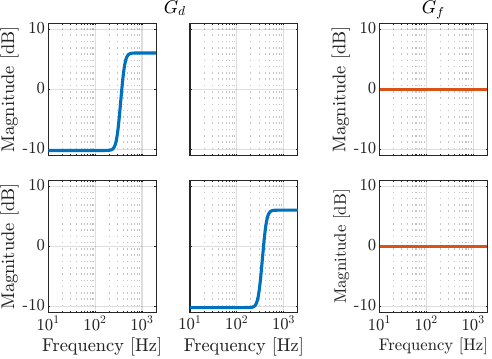}
    \caption{Disturbance model $G_d$ \tikzline{MatlabBlue} and fault model $G_f$ \tikzline{MatlabRed}.}
    \label{chpt2:fig:FFR_GdGf}
\end{figure}

\begin{figure}[t]
    \centering
    \includegraphics[scale=0.845]{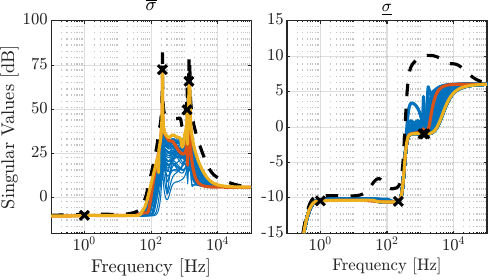}
    \caption{Singular value plots of uncertain transfer function matrix $\tilde{G}_d(s,\Delta)$ in $\tikzline{MatlabBlue}$ together with the singular values of worst-case gain upper bound $\mathcal{F}_{u} \left( D(s), \Delta_{wc}(s) \right)$ in $\tikzline{MatlabRed}$, as well as the multiple frequency interpolant via the BNP approach $\bar{G}_{d,\mathrm{init}}(s) = \mathcal{F}_{u} \left( D(s), \Delta_{np}(s) \right)$ $\tikzline{MatlabYellow}$. The frequencies over which is interpolated are indicated $\tikzcross{Black}$. The singular values of the weighted upper bound are indicated in \tikzdashedline{black}.}
    \label{chpt2:fig:FFR_svd_Gdbar}
\end{figure}

The singular values in Figure \ref{chpt2:fig:FFR_svd_Gdbar} show that the BNP interpolant provides a good initial model, slightly inflated to satisfy Lemma \ref{chpt2:definition:untight}. Figure \ref{chpt2:fig:FFR_svd_transfers} shows $\bar{\sigma}(T_{\varepsilon f})$ and the optimal fault sensitivity for $\Delta = 0$, demonstrating limited degradation due to model uncertainty. The optimal post filter for $\Delta = 0$, shown in \ref{chpt2:fig:FFR_Q}, matches the inverse of the disturbance model. Comparing $\bar{G}_{d}(\Delta)$ in Figure \ref{chpt2:fig:FFR_svd_Gdbar} highlights the impact of modeling uncertainty. Consequently, the singular values of $T_{\epsilon \Tilde{d}} : [ r \:\: d ] \rightarrow \epsilon$ 
are mapped below $\gamma=1$, satisfying Theorem \ref{chpt2:theorem:main_solution}, while achieving amplification of $T_{\epsilon f} : f \rightarrow \epsilon$. 

\newpage

\vspace{-7mm}
\begin{figure}[t!]
    \centering
    \includegraphics[scale=0.81]{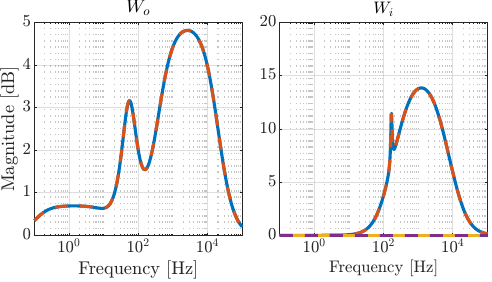}
    \vspace{-3mm}
    \caption{Bode magnitude plots of the iteratively tuned filters $W_{o}(s)$ and $W_{i}(s)$. The first and second entry of the block diagonal $W_{o}$ are depicted as \tikzline{MatlabBlue} and \tikzdashedline{MatlabRed}. The first, second, third, and fourth entry of these $W_i$ are depicted as \tikzline{MatlabBlue}, \tikzdashedline{MatlabRed}, \tikzline{MatlabYellow}, and \tikzdashedline{MatlabPurple}, respectively.}
    \label{chpt2:fig:FFR_Walpha_beta}
\end{figure}
\vspace{-6mm}
\begin{figure}[t!]
    \centering
    \includegraphics[scale=0.81]{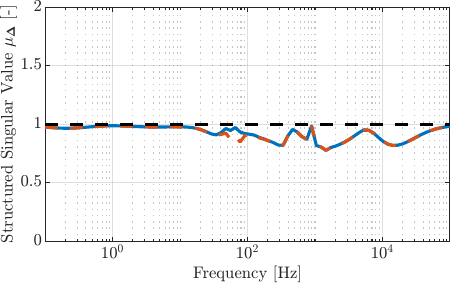}
    \vspace{-3mm}
    \caption{Structured singular value $\mu_{\mathbf{\Delta}} (\bar{G}_{do}^{-1}(j\omega) \tilde{G}_d(j\omega,\Delta))$. Since the lower bound \tikzline{MatlabBlue} and  the upper bound \tikzline{MatlabRed} of the structured singular value is below $\gamma = 1$ the user-defined maximum disturbance sensitivity is satisfied.}
    \label{chpt2:fig:SSV}
\end{figure}
\vspace{-6mm}
\begin{figure}[t!]
    \centering
    \includegraphics[scale=0.81]{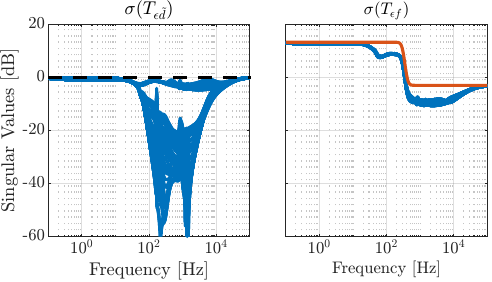}
    \vspace{-3mm}
    \caption{Singular values of transfer $T_{\epsilon \Tilde{d}} : \begin{bmatrix}
        r & d
    \end{bmatrix} \rightarrow \epsilon$ in $\tikzline{MatlabBlue}$ together with bound $\gamma = 1$ in $\tikzdashedline{MatlabBlack}$ (left) and the singular values of transfer $T_{\epsilon f} : f \rightarrow \epsilon$ (right) for 50 different uncertainties $\Delta \in \mathbf{\Delta}$. On top of that, the singular values of the maximum achievable transfer $T_{\epsilon f}$ is plotted in $\tikzline{MatlabRed}$ when the optimal post-filter is used in the case of no model uncertainty. Hence, performance degradation due to model uncertainty is limited.}
    \label{chpt2:fig:FFR_svd_transfers}
\end{figure}
\vspace{-6mm}
\begin{figure}[t!]
    \centering
    \includegraphics[scale=0.81]{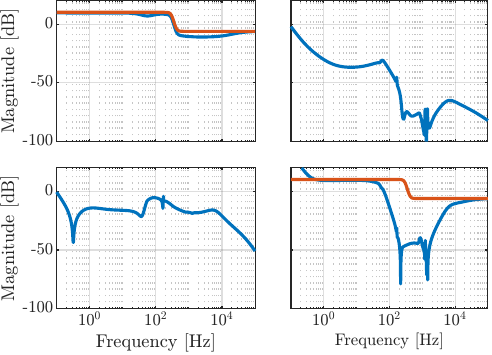}
    \vspace{-3mm}
    \caption{Bode magnitude plot of the post-filter $R$ in $\tikzline{MatlabBlue}$ together with the optimal post-filter if no model uncertainty is present in $\tikzline{MatlabRed}$. Observe the shape of the post-filter with respect to the singular values of $\bar{G}_{d}$ in Figure \ref{chpt2:fig:FFR_svd_Gdbar}.}
    \label{chpt2:fig:FFR_Q}
\end{figure}
\vspace{-6mm}

\newpage \subsection{Time-domain response} 
A time domain simulation demonstrates fault detection using $R$ from Figure \ref{chpt2:fig:FFR_Q}. The disturbance is modeled as independent normally distributed white noise ($\sigma = 0.1$), and a setpoint $r$ (block-signal, amplitude 1, $f_r = 50$ Hz) is applied at $T=0.05$ s. A fault with amplitude 1 occurs at $T=0.1s$. It consists of two block-shaped components and a $200$ Hz component, a frequency range characterized by significant model uncertainty. Residual $\epsilon_1$ and $\epsilon_2$ for different $\Delta \in \mathbf{\Delta}$ are shown in Figure \ref{chpt2:fig:FFR_epsilon} alongside $r$ and $f$. 

Disturbances and setpoints are suppressed, and faults are distinct in the residuals. Although deriving time-domain bounds is outside this scope, a conservative detection threshold of $\abs{\beta} = 1$, based on $t \in [0.05,0.1]$ s, enables fault detection despite disturbances and modeling uncertainty.

\begin{figure}[t]
    \centering
    \includegraphics[scale=0.845]{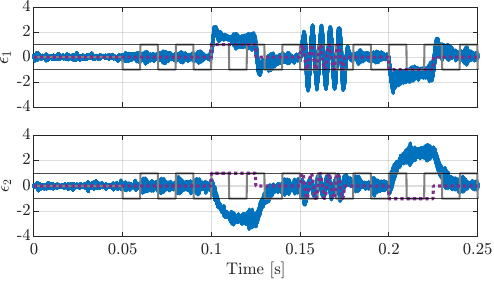}
    \caption{Time response of both residual signals $\epsilon_1$ and $\epsilon_2$ in $\tikzline{MatlabBlue}$ for 50 different uncertainties $\Delta \in \mathbf{\Delta}$ together with the reference input $r$ in $\tikzline{MatlabGray50}$ and fault input $f$ in $\tikzdottedline{MatlabPurple}$. A detection threshold $\abs{\beta} = 1$ is plotted in $\tikzline{MatlabBlack}$. Indeed, the residuals exceed the user-defined $\beta$-bounds when a fault is present in the system.}
    \label{chpt2:fig:FFR_epsilon}
\end{figure}

\section{Conclusion}
In this article, an optimal solution is derived which solves the robust fault detection filter design problem for continuous-time LTI uncertain systems operating in an open or closed-loop setting. The solution is obtained by solving a single Riccati equation, which maximizes performance across all frequencies, and is completely determined by an upper-bound envelope which bounds the uncertainty and disturbance models. By achieving an optimal compromise between sensitivity to faults and the rejection of disturbances and modeling uncertainty inherent in any practical systems, the proposed method effectively addresses a key challenge in fault-diagnosis system design.

    \section*{CRediT authorship contribution statement}
    \textbf{Koen Classens:} Conceptualization, Methodology, Software, Validation, Investigation, Visualization, Writing - Original Draft, Writing - Review \& Editing, \textbf{Tjeerd Ickenroth:} Conceptualization, Methodology, Software, Validation, Investigation, Writing - Review \& Editing, \textbf{Jeroen van de Wijdeven:} Conceptualization, Writing - Review \& Editing, Supervision, 
    \textbf{Tom Oomen:} Conceptualization, Writing - Review \& Editing, Supervision, Funding acquisition

    \section*{Declaration of competing interest}
    The authors declare that they have no known competing financial interests or personal relationships that could have appeared to influence the work reported in this paper.
    
    \section*{Acknowledgements}
    This work is supported by ASML, Veldhoven, the Netherlands. In addition, this work is funded by Topconsortia voor Kennis en Innovatie (TKI). Paul Tacx is gratefully acknowledged for providing the model of the experimental reticle stage.
    
    \bibliographystyle{elsarticle-num} 
    \bibliography{refs}

@Book{balasRobustControlToolbox2024,
  author = {Balas, Gary and Chiang, Richard and Packard, Andy and Safonov, Michael},
  title  = {Robust control toolbox user's guide},
  year   = {2024},
}

@Book{ballInterpolationRationalMatrix1990,
  author    = {Ball, J. A. and Gohberg, I. and Rodman, L.},
  title     = {Interpolation of {{Rational Matrix Functions}}},
  doi       = {10.1007/978-3-0348-7709-1},
  location  = {Basel},
  publisher = {Birkhäuser Basel},
  year      = {1990},
}

@Article{braatzComputationalComplexitySpl1994,
  author       = {Braatz, R.P. and Young, P.M. and Doyle, J.C. and Morari, M.},
  journal = {IEEE Trans. Autom. Control},
  title        = {Computational Complexity of $\mu$ Calculation},
  doi          = {10.1109/9.284879},
  issn         = {1558-2523},
  number       = {5},
  pages        = {1000--1002},
  urldate      = {2024-05-06},
  volume       = {39},
  eventtitle   = {{{IEEE Transactions}} on {{Automatic Control}}},
  year         = {1994},
}

@Book{chenRobustModelBasedFault1999,
  author    = {Chen, J. and Patton, R. J.},
  title     = {Robust {{Model-Based Fault Diagnosis}} for {{Dynamic Systems}}},
  doi       = {10.1016/s0005-1098(01)00290-4},
  isbn      = {978-1-4613-7344-5},
  location  = {{Boston, MA}},
  pages     = {1089--1091},
  publisher = {{Springer US}},
  booktitle = {Automatica},
  issn      = {00051098},
  year      = {1999},
}

@InProceedings{Classens2021,
  author    = {Classens, K. and Heemels, W. P. M. H. and Oomen, T.},
  booktitle = {2021 IEEE American Control Conference (ACC)},
  title     = {{Closed-loop Aspects in MIMO Fault Diagnosis with Application to Precision Mechatronics}},
  doi       = {10.23919/ACC50511.2021.9482785},
  isbn      = {9781665441971},
  location  = {New Orleans, LA, USA},
  pages     = {1756--1761},
  issn      = {07431619},
  year      = {2021},
}

@InProceedings{Classens2021DigitalTwin,
  author    = {Classens, K. and Heemels, W. P. M. H. and Oomen, T.},
  booktitle = {2021 IEEE 1st International Conference on Digital Twins and Parallel Intelligence (DTPI)},
  title     = {{Digital Twins in Mechatronics: From Model-based Control to Predictive Maintenance}},
  doi       = {10.1109/DTPI52967.2021.9540144},
  isbn      = {9781665433372},
  location  = {Beijing, China},
  pages     = {336--339},
  year      = {2021},
}

@Article{classens2023Opportunities,
  author    = {Classens, K. and van de Wijdeven, J. and Heemels, W. P. M. H. and Oomen, T.},
  title     = {Opportunities of Digital Twins for High-tech Systems: From Fault Diagnosis and Predictive Maintenance to Control Reconfiguration},
  number    = {5},
  pages     = {5--12},
  volume    = {63},
  journal   = {Mikroniek},
  publisher = {DSPE},
  year      = {2023},
}

@thesis{classensFaultDiagnosisUncertain2024,
  type = {Phd Thesis 1 (Research TU/e / Graduation TU/e)},
  title = {Fault {{Diagnosis}} for {{Uncertain Systems}} in {{Closed Loop}}: {{Applied}} to {{Semiconductor Equipment}}},
  author = {Classens, Koen},
  year = {2024},
  institution = {Eindhoven University of Technology},
  location = {Eindhoven, The Netherlands},
  langid = {english}
}

@Book{dingModelbasedFaultDiagnosis2008,
  author     = {Ding, Steven X.},
  date       = {2008},
  title      = {Model-Based Fault Diagnosis Techniques: Design Schemes, Algorithms, and Tools},
  doi        = {10.1007/978-3-540-76304-8},
  isbn       = {9783540763031},
  location   = {Berlin},
  pages      = {1--473},
  pagetotal  = {473},
  publisher  = {Springer Nature},
  booktitle  = {Model-based Fault Diagnosis Techniques: Design Schemes, Algorithms, and Tools},
  shorttitle = {Model-Based Fault Diagnosis Techniques},
  year       = {2008},
}

@Book{dingDatadrivenDesignFault2014,
  author    = {Ding, S. X.},
  title     = {Data-Driven {{Design}} of {{Fault Diagnosis}} and {{Fault-tolerant Control Systems}}},
  doi       = {10.1007/978-1-4471-6410-4},
  location  = {{London}},
  publisher = {{Springer London}},
  year      = {2014},
}

@Article{Ding2000,
  author   = {Ding, S. X. and Jeinsch, T. and Frank, P. M. and Ding, E. L.},
  title    = {{A unified approach to the optimization of fault detection systems}},
  pages    = {725--745},
  doi      = {10.1002/1099-1115(200011)14:7<725::AID-ACS618>3.0.CO;2-Q},
  journal  = {Int. J. Adapt. Control Signal Process.},
  year     = {2000},
}

@Article{dingFaultDetectionFactorization1990,
  author       = {Ding, Xianchun and Frank, Paul M.},
  date         = {1990-06},
  journal = {Syst. Control Lett.},
  title        = {Fault Detection via Factorization Approach},
  doi          = {10.1016/0167-6911(90)90094-B},
  issn         = {01676911},
  number       = {5},
  pages        = {431--436},
  volume       = {14},
  journal      = {Systems {\&} Control letters},
  langid       = {english},
  year         = {1990},
}

@InProceedings{wei2025Application,
  author       = {Wei, Pengyu and Wu, Shihao and Gao, Dianchun and Hao, Zhongyang and Song, Fazhi},
  booktitle    = {2025 44th Chinese Control Conference (CCC)},
  title        = {Application of a Robust Fault Detection Method for Processing Modeling Uncertainty in LPV system},
  doi          = {10.23919/CCC64809.2025.11178816},
  organization = {IEEE},
  pages        = {5285--5290},
  year         = {2025},
}

@Article{frank1997survey,
  author  = {Frank, P. M. and Ding, X.},
  title   = {{Survey of robust residual generation and evaluation methods in observer-based fault detection systems}},
  doi     = {10.1016/S0301-4770(08)60756-3},
  issn    = {03014770},
  pages   = {403--424},
  volume  = {7},
  journal = {J. Process Control},
  year    = {1997},
}

@InProceedings{venkataraman2016robust,
  author     = {Venkataraman, R. and Seiler, P.},
  booktitle  = {2016 {{American Control Conference}} ({{ACC}})},
  doi          = {10.1109/ACC.2016.7526079},
  title      = {Robust LPV estimator synthesis using integral quadratic constraints},
  pages      = {4611-4616},
  year       = {2016},
}

@Book{gertlerFaultDetectionDiagnosis1998,
  author   = {Gertler, J.},
  title    = {Fault {{Detection}} and {{Diagnosis}} in {{Engineering Systems}}},
  location = {{Boca Raton}},
  doi      = {10.1201/9780203756126},
  year     = {1998},
}

@Article{Glover2011,
  author  = {Glover, Keith and Varga, Andras},
  title   = {{On solving non-standard $\mathcal{H}_-/\mathcal{H}_{2/\infty}$ fault detection problems}},
  doi     = {10.1109/CDC.2011.6160723},
  issn    = {25762370},
  number  = {2},
  pages   = {891--896},
  isbn    = {9781612848006},
  journal = {Proceedings of the IEEE Conference on Decision and Control},
  year    = {2011},
}

@Article{henryTheoriesDesignAnalysis2021,
  author       = {Henry, D.},
  journal = {J. Franklin Inst.},
  title        = {Theories for Design and Analysis of Robust  {{$H_\infty$}}/{{$H_-$}}  Fault Detectors},
  doi          = {10.1016/j.jfranklin.2020.11.006},
  issn         = {00160032},
  number       = {1},
  pages        = {1152--1183},
  urldate      = {2021-03-04},
  volume       = {358},
  year         = {2021},
}

@Article{hollandDevelopmentSkewLower2005,
  author       = {Holland, R. and Young, P. and Zhu, C.},
  journal = {Int. J. Robust Nonlinear Control},
  title        = {Development of a Skew $\mu$ Lower Bound},
  doi          = {10.1002/rnc.1003},
  issn         = {1099-1239},
  number       = {11},
  pages        = {495--506},
  urldate      = {2024-05-06},
  volume       = {15},
  langid       = {english},
  year         = {2005},
}

@Article{hollandDevelopmentSkewUpper2005,
  author       = {Holland, R. and Young, P. and Zhu, C.},
  journal = {Int. J. Robust Nonlinear Control},
  title        = {Development of a Skew $\mu$ Upper Bound},
  doi          = {10.1002/rnc.1028},
  issn         = {1099-1239},
  number       = {18},
  pages        = {905--921},
  urldate      = {2024-05-06},
  volume       = {15},
  langid       = {english},
  year         = {2005},
}

@InProceedings{houLMIApproachHinfty1996,
  author    = {Hou, M. and Patton, R. J.},
  booktitle = {UKACC International Conference on Control},
  title     = {An {{LMI}} Approach to {{$H_-$}}/{{$H_{\infty}$}} Fault Detection Observers},
  doi       = {10.1049/cp:19960570},
  pages     = {305-310},
  issn      = {0537-9989},
  year      = {1996},
}

@Book{isermannFaultdiagnosisSystemsIntroduction2006,
  author     = {Isermann, R.},
  title      = {Fault-Diagnosis Systems: An Introduction from Fault Detection to Fault Tolerance},
  isbn       = {978-3-540-24112-6},
  location   = {{Berlin; New York}},
  pagetotal  = {475},
  publisher  = {{Springer}},
  doi        = {10.1007/3-540-30368-5},
  langid     = {english},
  shorttitle = {Fault-Diagnosis Systems},
  year       = {2006},
}

@Article{jaimoukhaMatrixFactorizationSolution2006,
  author       = {Jaimoukha, I. M. and Li, Z. and Papakos, V.},
  journal = {Automatica},
  title        = {A Matrix Factorization Solution to the {{$H_{-}$}}/{{$H_{\infty}$}} Fault Detection Problem},
  doi          = {10.1016/j.automatica.2006.06.002},
  number       = {11},
  volume       = {42},
  year         = {2006},
}

@Article{Li2008,
  author  = {Li, Jingjing and Zhou, Kemin and Ren, Zhang},
  title   = {{Robust fault diagnosis for linear time invariant uncertain systems}},
  doi     = {10.1109/WCICA.2008.4593089},
  pages   = {1170--1173},
  isbn    = {9781424421145},
  journal = {Proceedings of the World Congress on Intelligent Control and Automation (WCICA)},
  year    = {2008},
}

@Article{liuLMIApproachMinimum2005,
  author       = {Liu, J. and Wang, J. L. and Yang, G.},
  journal = {Automatica},
  title        = {An {{LMI}} Approach to Minimum Sensitivity Analysis with Application to Fault Detection},
  doi          = {10.1016/j.automatica.2005.06.005},
  issn         = {00051098},
  number       = {11},
  pages        = {1995--2004},
  urldate      = {2024-05-06},
  volume       = {41},
  langid       = {english},
  year         = {2005},
}

@InProceedings{liuOptimalSolutionsMultiobjective2007,
  author     = {Liu, Nike and Zhou, Kemin},
  booktitle  = {2007 46th {{IEEE Conference}} on {{Decision}} and {{Control}}},
  date       = {2007-12},
  title      = {Optimal Solutions to Multi-Objective Robust Fault Detection Problems},
  doi        = {10.1109/CDC.2007.4434123},
  eventtitle = {2007 46th {{IEEE Conference}} on {{Decision}} and {{Control}}},
  isbn       = {1424414989},
  pages      = {981--988},
  issn       = {25762370},
  journal    = {Proceedings of the IEEE Conference on Decision and Control},
  year       = {2007},
}

@Article{marzatModelbasedFaultDiagnosis2012,
  author     = {Marzat, J. and Piet-Lahanier, H. and Damongeot, F. and Walter, E.},
  title      = {Model-Based Fault Diagnosis for Aerospace Systems: A Survey},
  doi        = {10.1177/0954410011421717},
  issn       = {0954-4100, 2041-3025},
  number     = {10},
  pages      = {1329--1360},
  volume     = {226},
  journal    = {Proc. Inst. Mech. Eng., Part G: J. Aerosp. Eng.},
  langid     = {english},
  shorttitle = {Model-Based Fault Diagnosis for Aerospace Systems},
  year       = {2012},
}

@Article{newlinGeneralizationStructuredSingular1998,
  author       = {Newlin, M. P. and Smith, R. S.},
  journal = {IEEE Trans. Autom. Contro},
  title        = {A Generalization of the Structured Singular Value and Its Application to Model Validation},
  doi          = {10.1109/9.701088},
  issn         = {1558-2523},
  number       = {7},
  pages        = {901--907},
  volume       = {43},
  eventtitle   = {{{IEEE Transactions}} on {{Automatic Control}}},
  year         = {1998},
}

@InProceedings{packardResultsWorstcasePerformance2000,
  author     = {Packard, A. and Balas, G. and Liu, R. and Shin, J.},
  booktitle  = {Proceedings of the 2000 {{American Control Conference}}. {{ACC}}},
  title      = {Results on Worst-Case Performance Assessment},
  doi        = {10.1109/ACC.2000.878616},
  eventtitle = {Proceedings of the 2000 {{American Control Conference}}. {{ACC}}},
  pages      = {2425-2427},
  urldate    = {2024-05-06},
  issn       = {0743-1619},
  year       = {2000},
}

@Article{packardComplexStructuredSingular1993,
  author       = {Packard, A. and Doyle, J.},
  journal = {Automatica},
  title        = {The Complex Structured Singular Value},
  doi          = {10.1016/0005-1098(93)90175-S},
  issn         = {0005-1098},
  number       = {1},
  pages        = {71--109},
  urldate      = {2022-01-14},
  volume       = {29},
  langid       = {english},
  year         = {1993},
}

@Article{patarticsWorstCaseUncertainty2023,
  author       = {Patartics, B. and Seiler, P. and Takarics, B. and Vanek, B.},
  journal = {IEEE Trans. Contr. Syst. Technol.},
  title        = {Worst {{Case Uncertainty Construction}} via {{Multifrequency Gain Maximization With Application}} to {{Flutter Control}}},
  doi          = {10.1109/TCST.2022.3173044},
  issn         = {1063-6536, 1558-0865, 2374-0159},
  number       = {1},
  pages        = {155--165},
  urldate      = {2024-04-26},
  volume       = {31},
  langid       = {english},
  year         = {2023},
}

@Article{ho2019robust,
  author       = {Ho, L. M.},
  journal = {IEEE Trans. Contr. Syst. Technol.},
  title        = {Robust Residual Generator Synthesis for Uncertain LPV Systems Applied to Lateral Vehicle Dynamics},
  number       = {3},
  doi        = {10.1109/TCST.2018.2789461},
  pages        = {1275-1283},
  volume       = {27},
  year         = {2019},
}

@InProceedings{patarticsConstructionUncertaintyMaximize2020,
  author     = {Patartics, B. and Seiler, P. and Vanek, B.},
  booktitle  = {2020 {{American Control Conference}} ({{ACC}})},
  title      = {Construction of an {{Uncertainty}} to {{Maximize}} the {{Gain}} at {{Multiple Frequencies}}},
  doi        = {10.23919/ACC45564.2020.9147542},
  eventtitle = {2020 {{American Control Conference}} ({{ACC}})},
  isbn       = {978-1-5386-8266-1},
  location   = {Denver, CO, USA},
  pages      = {2643--2648},
  urldate    = {2024-04-26},
  langid     = {english},
  year       = {2020},
}

@InProceedings{roosSystemsModelingAnalysis2013,
  author     = {Roos, Clement},
  booktitle  = {2013 {{IEEE Conference}} on {{Computer Aided Control System Design}} ({{CACSD}})},
  title      = {Systems Modeling, Analysis and Control ({{SMAC}}) Toolbox: {{An}} Insight into the Robustness Analysis Library},
  isbn       = {978-1-4799-1565-1},
  doi = {10.1109/CACSD.2013.6663479},
  location   = {Hyderabad, India},
  pages      = {176--181},
  publisher  = {IEEE},
  langid     = {english},
  shorttitle = {Systems Modeling, Analysis and Control ({{SMAC}}) Toolbox},
  year       = {2013},
}

@InProceedings{sadrniaRobustInftyMu1997,
  author     = {Sadrnia, M. A. and Patton, R. J. and Chen, J.},
  booktitle  = {1997 {{European Control Conference}} ({{ECC}})},
  title      = {Robust {{$H_-$}}/$\mu$ Fault Diagnosis Observer Design},
  doi        = {10.23919/ECC.1997.7082314},
  eventtitle = {1997 {{European Control Conference}} ({{ECC}})},
  pages      = {1502--1507},
  year       = {1997},
}

@Article{huang2026towards,
  author  = {Huang, Jiacong and Song, Fazhi and Gao, Dianchun and Tan, Jiubin},
  title   = {Towards predictive maintenance of lithography systems: Robust fault diagnosis via {LPV}-to-{LTI} reformulation},
  doi     = {10.1016/j.rineng.2026.109172},
  pages   = {109172},
  volume  = {29},
  journal = {Results Eng.},
  year    = {2026},
}

@Book{skogestadMultivariableFeedbackControl2005,
  author    = {Skogestad, Sigurd and Postlethwaite, Ian},
  title     = {Multivariable {{Feedback Control Analysis}} and {{Design}}},
  publisher = {{John Wiley \& Sons}},
  year      = {2005},
}

@Article{stoustrupFaultEstimationStandard2002,
  author       = {Stoustrup, J. and Niemann, H. H.},
  journal = {Int. J. Robust Nonlinear Control},
  title        = {Fault Estimation - a Standard Problem Approach},
  doi          = {10.1002/rnc.716},
  issn         = {1049-8923, 1099-1239},
  number       = {8},
  pages        = {649--673},
  urldate      = {2023-01-25},
  volume       = {12},
  langid       = {english},
  shorttitle   = {Fault Estimation?},
  year         = {2002},
}

@InProceedings{Tacx2022,
  author    = {Tacx, Paul and Oomen, Tom},
  title     = {{A One-step Approach for Centralized Overactuated Motion Control of a Prototype Reticle Stage}},
  doi       = {10.1016/j.ifacol.2022.11.202},
  booktitle = {2nd Modeling, Estimation and Control Conference},
  pages     = {308--313},
  year      = {2022},
}

@InProceedings{fengtaoFaultDetectionObserver2005,
  author     = {Tao, F. and Zhao, Q},
  booktitle  = {Proceedings of 2005 {{IEEE Conference}} on {{Control Applications}}},
  title      = {Fault Detection Observer Design with Unknown Inputs},
  doi        = {10.1109/CCA.2005.1507307},
  eventtitle = {2005 {{IEEE Conference}} on {{Control Applications}}, 2005. {{CCA}} 2005.},
  isbn       = {978-0-7803-9354-7},
  location   = {Toronto, Canada},
  pages      = {1275--1280},
  urldate    = {2024-06-02},
  langid     = {english},
  year       = {2005},
}

@Article{Varga2017,
  author  = {Varga, Andreas},
  title   = {{Fault Detection and Isolation Tools (FDITOOLS) User's Guide}},
  arxivid = {1703.08480},
  journal = {},
  doi     = {10.48550/arXiv.1703.08480},
  year    = {2017},
}

@InProceedings{haibowangIterativeLMIApproach2003,
  author    = {Wang, H. and Wang, J. and Liu, J. and Lam, J.},
  booktitle = {42nd {{IEEE International Conference}} on {{Decision}} and {{Control}}},
  title     = {Iterative {{LMI}} Approach for Robust Fault Detection Observer Design},
  doi       = {10.1109/CDC.2003.1272905},
  isbn      = {978-0-7803-7924-4},
  location  = {Maui, Hawaii, USA},
  pages     = {1974--1979},
  urldate   = {2024-05-06},
  langid    = {english},
  year      = {2003},
}

@Article{wangLMIApproachIndex2007,
  author       = {Wang, J. L. and Yang, G. and Liu, J.},
  journal = {Automatica},
  title        = {An {{LMI}} Approach to  {{$H_-$}}  Index and Mixed {{$H_-$}}/{{$H_{\infty}$}} Fault Detection Observer Design},
  doi          = {10.1016/j.automatica.2007.02.019},
  issn         = {00051098},
  number       = {9},
  pages        = {1656--1665},
  urldate      = {2021-03-04},
  volume       = {43},
  year         = {2007},
}

@Article{zhongLMIApproachDesign2003,
  author       = {Zhong, M. and Ding, S. X. and Lam, J. and Wang, H.},
  journal = {Automatica},
  title        = {An {{LMI}} Approach to Design Robust Fault Detection Filter for Uncertain {{LTI}} Systems},
  doi = {10.1016/S0005-1098(02)00269-8},
  langid       = {english},
  year         = {2003},
}

@Book{Zhou1996a,
  author    = {Zhou, K. and Doyle, J. and Glover, K.},
  title     = {{Robust and optimal control}},
  doi       = {10.1016/s0005-1098(97)00132-5},
  publisher = {Prentice Hall},
  issn      = {01912216},
  year      = {1996},
}

@Book{zolghadriFaultDiagnosisFaultTolerant2014,
  author     = {Zolghadri, A. and Henry, D. and Cieslak, J. and Efimov, D. and Goupil, P.},
  title      = {Fault {{Diagnosis}} and {{Fault-Tolerant Control}} and {{Guidance}} for {{Aerospace Vehicles}}: {{From Theory}} to {{Application}}},
  doi        = {10.1007/978-1-4471-5313-9},
  location   = {London},
  publisher  = {Springer},
  shorttitle = {Fault {{Diagnosis}} and {{Fault-Tolerant Control}} and {{Guidance}} for {{Aerospace Vehicles}}},
  year       = {2014},
}

\appendix
\section{Alternative representation residual dynamics for uncertain closed-loop systems}\label{chpt2:Appendix:optimal_ratio}
Consider the residual dynamics in \eqref{chpt2:residual_expression_long}. This appendix shows how to rewrite the residual dynamics to a different insightful form by factoring out a common uncertainty. Factoring out $G_{d}$ and $G_{f}$ gives
\begin{align*}
    \epsilon &= R\Tilde{M}_u\{ \tilde{G}_u(\Delta) K S_\Delta r + (I - \tilde{G}_u(\Delta)KS_\Delta) G_d(\Delta) d + (I - \tilde{G}_u(\Delta)KS_\Delta) G_f(\Delta) f \} \nonumber.
\end{align*}
Taking out $(I - \tilde{G}_u(\Delta)KS_\Delta)$ gives
\begin{align}
    \epsilon &= R\Tilde{M}_u (I - \tilde{G}_u(\Delta)KS_\Delta) \{ (I - \tilde{G}_u(\Delta)KS_\Delta)^{-1} \tilde{G}_u(\Delta) K S_\Delta r \label{chpt2:eq:app_res_dyn_1} + G_d(\Delta) d + G_f(\Delta) f \}.
\end{align}
Consider the following related to the factor $I - \tilde{G}_u(\Delta)KS_\Delta$. Using that $S_\Delta+T_\Delta = I$ in \eqref{chpt2:eq:app_res_dyn_1}, with $T_\Delta = (I + G_u(\Delta)K)^{-1} G_u(\Delta)$ $K = G_u(\Delta)K (I + G_u(\Delta)K)^{-1}$, this factor can be simplified as follows.
\begin{align*}
I - \tilde{G}_u(\Delta)KS_\Delta &= I - (G_u(\Delta) - G_u(0))KS_\Delta\\
    &= I-T_\Delta+G_u(0)KS_\Delta\\
    &= S_\Delta+G_u(0)KS_\Delta\\
    &= (I+G_u(0)K)S_\Delta\\
    &= S^{-1}S_\Delta,
\end{align*}
where $S = (I+G(0)K)^{-1}$ is the nominal sensitivity function. Next, consider the following related to the factor $(I - \tilde{G}_u(\Delta)KS_\Delta)^{-1} \tilde{G}_u(\Delta) K S_\Delta$ in \eqref{chpt2:eq:app_res_dyn_1}. \allowdisplaybreaks
\begin{align*}
    (I - \tilde{G}_u(\Delta)KS_\Delta)^{-1}\tilde{G}_u(\Delta)KS_\Delta 
    &= \left(\tilde{G}_u(\Delta)KS_\Delta\left((\tilde{G}_u(\Delta)KS_\Delta)^{-1} -I\right)\right)^{-1} \tilde{G}_u(\Delta)KS_\Delta\\
    &= \left((\tilde{G}_u(\Delta)KS_\Delta)^{-1}-I\right)^{-1} (\tilde{G}_u(\Delta)KS_\Delta)^{-1} \tilde{G}_u(\Delta)KS_\Delta \\
    &= \left(S_\Delta^{-1}(\tilde{G}_u(\Delta) K)^{-1}-I\right)^{-1}\\
    &= \left((I+G_u(\Delta)K)(\tilde{G}_u(\Delta)K)^{-1}-I\right)^{-1}\\
    &= \left((I+G_u(\Delta)K-\tilde{G}_u(\Delta)K)(\tilde{G}_u(\Delta)K)^{-1}\right)^{-1}\\
    & \text{Using that $\tilde{G}_u(\Delta) = G_u(\Delta)-G_u(0)$ gives}\\
    &= \left((I+G_u(0)K)(\tilde{G}_u(\Delta)K)^{-1}\right)^{-1}\\
    &= \tilde{G}_u(\Delta)KS,
\end{align*}
Hence, the residual dynamics \eqref{chpt2:eq:app_res_dyn_1}, and thus \eqref{chpt2:residual_expression_long}, can be written as
\begin{equation*}
    \begin{split}
    \epsilon &= R\tilde{M}_u S^{-1}S_\Delta\Big( \tilde{G}_u(\Delta)CS r + G_d(\Delta) d + G_f(\Delta) f \Big),
    \end{split}
\end{equation*}
which reveals that a large part of the uncertainty can be factored out and applies to all exogenous inputs $r$, $d$, and $f$ equally.

\section{Proof of Lemma \ref{chpt2:definition:untight}}
\label{app:proof:definition:untight}

\begin{proof}
    If 
    \begin{equation*}
        \bar{\sigma}( \bar{G}_{do}^{-1}(j\omega) \tilde{G}_d(j\omega,\Delta) ) \leq 1, \quad \forall \omega, \forall \Delta \in \mathbf{\Delta},
    \end{equation*}
    then 
    \begin{equation*}
    \begin{split}
        \sup_{\norm{d_{1} (j \omega)}_{2} = 1} \norm{ \bar{G}_{do}^{-1}(j\omega) \tilde{G}_d(j\omega,\Delta) d_{1} (j \omega)}_{2} \leq 1, \quad \forall \omega, \forall \Delta \in \mathbf{\Delta}.
    \end{split}
    \end{equation*}
    Since any $d_{2}$ satisfies $\|\Tilde{d}_2\|_2 = 1$,
    \begin{equation*}
    \begin{split}
        \sup_{\norm{d_{1} (j \omega)}_{2} = 1} \norm{ \bar{G}_{do}^{-1}(j\omega) \tilde{G}_d(j\omega,\Delta) d_{1} (j \omega)}_{2} \leq \|\Tilde{d}_2 (j\omega)\|_2, \quad \forall \|\Tilde{d}_2 (j\omega)\|_2 = 1, \forall \omega, \forall \Delta \in \mathbf{\Delta}.
    \end{split}
    \end{equation*}
    Using the inner of $\bar{G}_{d}(j\omega)$, 
    \begin{equation*}
    \begin{split}
        \sup_{\norm{d_{1} (j \omega)}_{2} = 1} \norm{ \bar{G}_{do}^{-1}(j\omega) \tilde{G}_d(j\omega,\Delta) d_{1} (j \omega)}_{2} \leq \|\bar{G}_{di}(j\omega) \Tilde{d}_2 (j\omega)\|_2, \quad \forall \|\Tilde{d}_2 (j\omega)\|_2 = 1, \forall \omega, \forall \Delta \in \mathbf{\Delta}.
    \end{split}
    \end{equation*}
    Hence, for all $\norm{d_{1} (j \omega)}_{2} = 1$, 
    \begin{equation*}
    \begin{split}
        \norm{ \bar{G}_{do}^{-1}(j\omega) \tilde{G}_d(j\omega,\Delta) d_{1} (j \omega)}_{2} \leq \|\bar{G}_{di}(j\omega) \Tilde{d}_2 (j\omega)\|_2, \quad \forall \|\Tilde{d}_2 (j\omega)\|_2 = 1,  \forall \omega, \forall \Delta \in \mathbf{\Delta}.
    \end{split}
    \end{equation*}
    Since for each $\Delta$, $d_{1} (j \omega)$ is such that $\tilde{G}_d(j\omega,\Delta) d_{1} (j \omega)$ points in the same direction as $\bar{G}_d(j\omega) d_{2} (j \omega)$, a scaling factor $g$ is introduced such that
    \begin{equation*}
    \begin{split}
        \tilde{G}_d(j\omega,\Delta) d_{1} (j \omega) = g \bar{G}_d(j\omega) d_{2} (j \omega), \quad \forall \|\Tilde{d}_2 (j\omega)\|_2 = 1, \forall \omega, \forall \Delta \in \mathbf{\Delta}
    \end{split}
    \end{equation*}
    with $g \leq 1$. Since these point in the same direction, clearly
    \begin{equation*}
    \begin{split}
        \bar{G}_{do}^{-1}(j\omega) \tilde{G}_d(j\omega,\Delta) d_{1} (j \omega) = g \bar{G}_{do}^{-1}(j\omega) \bar{G}_d(j\omega) d_{2} (j \omega), \quad \forall \|\Tilde{d}_2 (j\omega)\|_2 = 1, \forall \omega, \forall \Delta \in \mathbf{\Delta}
    \end{split}
    \end{equation*}
    point in the same direction. Thus,
    \begin{equation*}
    \begin{split}
        \norm{ \tilde{G}_d(j\omega,\Delta) d_{1} (j \omega)}_{2} \leq \|\bar{G} (j\omega) \Tilde{d}_2 (j\omega)\|_2, \quad \forall \|\Tilde{d}_2 (j\omega)\|_2 = 1,  \forall \omega, \forall \Delta \in \mathbf{\Delta}.
    \end{split}
    \end{equation*}
\end{proof}

\section{Proof of Lemma \ref{chpt2:definition:tight}}
\label{app:proof:definition:tight}

\begin{proof}
    Start from the observation that if all the worst-case singular values $i = 1, \ldots, n_{\sigma}$ are equal to one for each $i = 1, \ldots, n_{\sigma}$,  then,
    \begin{equation*}
        \sup_{\Delta \in \mathbf{\Delta}} \bar{\sigma}_{i} ( \bar{G}_{do}^{-1}(j\omega) \tilde{G}_d(j\omega,\Delta) ) = 1, \quad \forall \omega, 
    \end{equation*}
    is equivalent to
    \begin{equation*}
    \begin{split}
        \sup_{\Delta \in \mathbf{\Delta}} \norm{ \bar{G}_{do}^{-1}(j\omega) \tilde{G}_d(j\omega,\Delta) d_{1} (j \omega)}_{2} = \|\bar{G}_{di}(j\omega) \Tilde{d}_2 (j\omega)\|_2, \quad \forall \|\Tilde{d}_2 (j\omega)\|_2 = 1,  \forall \omega. 
    \end{split}
    \end{equation*}
    where $\tilde{d}_1$ is a unitary vector such that the output of the particular realization of
$\tilde{G}_d(j\omega,\Delta)\tilde{d}_1$ matches the direction of $\bar{G}_d(j\omega)\tilde{d}_2$. Hence, $\bar{G}_{do}^{-1}(j\omega) \tilde{G}_d(j\omega,\Delta) d_{1} (j \omega)$ points in the same direction as $\bar{G}_{do}^{-1}(j\omega) \bar{G}_d(j\omega) d_{2} (j \omega) = \bar{G}_{di}(j\omega) d_{2} (j \omega)$. Since these point in the same direction, clearly
    \begin{equation*}
    \begin{split}
       \sup_{\Delta \in \mathbf{\Delta}} \left( \bar{G}_{do}^{-1}(j\omega) \tilde{G}_d(j\omega,\Delta) d_{1} (j \omega) \right) = \bar{G}_{di}(j\omega) d_{2} (j \omega), \quad \forall \|\Tilde{d}_2 (j\omega)\|_2 = 1, \forall \omega. 
    \end{split}
    \end{equation*}
Then, premultiplication with $\bar{G}_{do}(j\omega)$ gives
    \begin{equation*}
    \begin{split}
       \sup_{\Delta \in \mathbf{\Delta}} \left( \tilde{G}_d(j\omega,\Delta) d_{1} (j \omega) \right) = \bar{G}(j\omega) d_{2} (j \omega), \quad \forall \|\Tilde{d}_2 (j\omega)\|_2 = 1, \forall \omega, 
    \end{split}
    \end{equation*}
    and therefore,
    \begin{equation*}
    \begin{split}
       \sup_{\Delta \in \mathbf{\Delta}} \norm{ \tilde{G}_d(j\omega,\Delta) d_{1} (j \omega) }_{2} = \norm{ \bar{G}(j\omega) d_{2} (j \omega) }_{2}, \quad \forall \|\Tilde{d}_2 (j\omega)\|_2 = 1, \forall \omega. 
    \end{split}
    \end{equation*}    
    Lemma \ref{chpt2:definition:tight} is proven by executing the steps in opposite direction, which concludes the proof.  
\end{proof}

\section{Proof of Theorem \ref{chpt2:theorem:main_solution}}
\label{app:proof:theorem:main_solution}

\begin{proof}
Since \eqref{chpt2:upperbound} holds where $\tilde{d}_1$ is a unitary vector such that the output of the particular realization of
$\tilde{G}_d(j\omega,\Delta)\tilde{d}_1$ matches the direction of $\bar{G}_d(j\omega)\tilde{d}_2$,
\begin{equation}
   \begin{split}
    \|R(j \omega) \tilde{G}_d(j\omega,\Delta)\tilde{d}_1\|_2 \leq \| R(j \omega) \bar{G}_d(j\omega)\tilde{d}_2\|_2 \quad \forall \omega, \forall \Delta \in \mathbf{\Delta},
   \end{split}
\end{equation}
Then, if $R(j\omega)$, $\tilde{G}_d(j\omega,\Delta)$ and $\bar{G}_d(j\omega)$ are stable proper transfer function matrices,
   \begin{equation*}
    \|R \tilde{G}_d(\Delta)\|_\infty \leq \|R \bar{G}_d\|_\infty \quad \forall \Delta \in \mathbf{\Delta}. \label{chpt2:upperbound_imply}
    \end{equation*}
It follows from substitution into \eqref{chpt2:performance_index} that
\begin{equation}\label{chpt2:performance_index_conservative}
    J_{i,\omega,\bar{\Delta}}(R) =
    \frac{\sigma_i\left( R(j\omega) T_{\tilde{\epsilon} f}(j\omega,\bar{\Delta}) \right)}{\|R \tilde{G}_d(\Delta)\|_\infty}
    \geq 
    \frac{\sigma_i\left( R(j\omega) T_{\tilde{\epsilon} f}(j\omega,\bar{\Delta}) \right)}{\|R \bar{G}_d \|_\infty}.
\end{equation}
Assuming a tight upper bound as in Assumption \ref{chpt2:assumption_4} in the theorem gives that
\begin{equation*}\label{chpt2:Gdbar_wc_eq}
    \|R \tilde{G}_d(\Delta)\|_\infty = \|R \bar{G}_d \|_\infty,
\end{equation*}
which results in equality in \eqref{chpt2:performance_index_conservative} as
\begin{equation}\label{chpt2:performance_index_optimal}
    J_{i,\omega,\bar{\Delta}}(R) = \frac{\sigma_i\left( R(j\omega) T_{\tilde{\epsilon} f}(j\omega,\bar{\Delta}) \right)}{\|R \tilde{G}_d(\Delta)\|_\infty}
    =
    \frac{\sigma_i\left( R(j\omega) T_{\tilde{\epsilon} f}(j\omega,\bar{\Delta}) \right)}{\|R \bar{G}_d \|_\infty}.
\end{equation}
From Assumptions \ref{chpt2:assumption_1} to \ref{chpt2:assumption_3} follows that a CIOF of $\bar{G}_d$ exists. I.e., $\bar{G}_d = G_{do} G_{di}$ exists, where $G_{di}$ is the co-inner matrix satisfying $G_{di}(j\omega)G_{di}^T(-j\omega)=I$ and having $\sigma(G_{di}(j\omega))=I$ for all $\omega$, and $G_{do}$ is co-outer satisfying $G_{do}^{-1} \in \mathcal{R}\mathcal{H}_\infty$ and thus $G_{do}^{-1}G_{do}=I$. Now parameterize $R$ as 
$$
R = \gamma Q G_{do}^{-1}
$$
where $Q \in$ \RHinf an arbitrary stable TFM. Substitution into \eqref{chpt2:performance_index_optimal} yields
\begin{align}
J_{i,\omega,\bar{\Delta}}(R) &= \nonumber
    \frac{\sigma_i\left( \gamma Q(j\omega) G_{do}^{-1}(j\omega) T_{\tilde{\epsilon} f}(j\omega,\bar{\Delta}) \right)}{\|\gamma  Q G_{di} \|_\infty}\\
    &= \frac{\sigma_i\left( Q(j\omega) G_{do}^{-1}(j\omega) T_{\tilde{\epsilon} f}(j\omega,\bar{\Delta}) \right)}{\|Q \|_\infty} \label{chpt2:R_subs}\\
    &\leq
    \sigma_i\left( G_{do}^{-1}(j\omega) T_{\tilde{\epsilon} f}(j\omega,\bar{\Delta}) \right), \nonumber
\end{align}
where it is used that $G_{di}(j\omega)G_{di}^\top(-j\omega)=I$ so that $\|Q G_{di} \|_\infty = \|Q \|_\infty$.

From the inequality in \eqref{chpt2:R_subs} follows that setting $R_{\mathrm{opt}} = \gamma G_{do}^{-1}$, i.e., $Q=I$ gives the optimal performance. Hence, it holds that for all $\omega, \Delta \in \mathbf{\Delta}$, and $i=1,\ldots, n_\sigma$
\begin{equation}\label{chpt2:perf_index_int}
        J_{i,\omega,\bar{\Delta}}(R_{\mathrm{opt}})
        =
        \sigma_i\left( G_{do}^{-1}(j\omega) T_{\tilde{\epsilon} f}(j\omega,\bar{\Delta}) \right),
\end{equation}
resulting in the optimal solution of Problem \ref{chpt2:prob:problem_Ding} for uncertain closed-loop systems. Evidently, this $R_{\mathrm{opt}}$ also maximizes the performance index with worst-case fault sensitivity \eqref{chpt2:performance_index}. The state-space realization of the outer term $G_{do}^{-1}$ follows directly from \cite[Theorem 13.35]{Zhou1996a}.
\end{proof}
    
\end{document}